\documentclass[lettersize,journal]{IEEEtran}
\IEEEoverridecommandlockouts
\usepackage{amsmath,amssymb,amsfonts,amsthm}
\usepackage{algorithmic}
\usepackage{algorithm}
\usepackage{array}
\usepackage{graphicx}
\usepackage{textcomp}
\usepackage{amsthm}
\usepackage{subfigure}
\usepackage{graphicx}
\usepackage{float}
\usepackage{enumerate}
\usepackage{tabularx}
\usepackage{stfloats}
\usepackage{url}
\usepackage{hyperref}
\usepackage{xcolor, soul, framed} 
\usepackage{xcolor}
\usepackage{verbatim}
\usepackage{pstricks}
\usepackage{booktabs} 
\usepackage{makecell}

\graphicspath{{figure/}} 

\newtheorem{thm}{Theorem}
\newtheorem{cor}{Corollary}

\newtheorem{lem}{Lemma}

\begin{document}
\renewcommand{\algorithmicrequire}{\textbf{Input:}} 
\renewcommand{\algorithmicensure}{\textbf{Output:}}
\title{Source Localization and Power Estimation through RISs: Performance Analysis and Prototype Validations \\  

\thanks{This work is supported by the National Natural Science Foundation of China under Grant 12141107, the Key Research and Development Program of Wuhan under Grant 2024050702030100, and the Interdisciplinary Research Program of HUST (2023JCYJ012). (\textit{Corresponding author: Tiebin Mi})}
}

\author{Fuhai Wang,~\IEEEmembership{Student Member,~IEEE,}
Tiebin Mi,~\IEEEmembership{Member,~IEEE,}
Chun Wang, 
Rujing Xiong,~\IEEEmembership{Student Member,~IEEE,}
Zhengyu Wang,~\IEEEmembership{Student Member,~IEEE,} 
and
Robert Caiming Qiu,~\IEEEmembership{Fellow,~IEEE} 
\thanks{Fuhai Wang is with the School of Electronic Information and Communications, Huazhong University of Science and Technology, Wuhan 430074, China, and also with the Institute of Artificial Intelligence, Huazhong University of Science and Technology, Wuhan 430074, China (e-mail: wangfuhai@hust.edu.cn).}
\thanks{Tiebin Mi, Chun Wang, Zhengyu Wang and Robert Caiming Qiu are with the School of Electronic Information and Communications, Huazhong University of Science and Technology, Wuhan 430074, China (e-mail: mitiebin@hust.edu.cn; m202272439@hust.edu.cn;  wangzhengyu@hust.edu.cn; caiming@hust.edu.cn).}
\thanks{Rujing Xiong is now with the School of Science and Engineering, The Chinese University of Hong Kong-Shenzhen, Shenzhen 518172, China. He was with the School of Electronic Information and Communications, Huazhong University of Science and Technology, Wuhan 430074, China (e-mail: rujingxiong@cuhk.edu.cn).}
}
\maketitle

\begin{abstract}
This paper investigates the capabilities and effectiveness of backward localization centered on reconfigurable intelligent surfaces (RISs). In the backward sensing paradigm, the region of interest (RoI) is illuminated using a set of diverse radiation patterns. These patterns encode spatial information into a sequence of measurements, which are subsequently processed to reconstruct the RoI. We show that a single RIS can estimate the direction of arrival of incident waves by leveraging configurational diversity, and that the spatial diversity provided by multiple RISs further improves the accuracy of source localization and power estimation. The underlying structure of the sensing operator in the multi-snapshot measurement process is clarified. For single-RIS localization, the sensing operator is decomposed into a product of structured matrices, each corresponding to a specific physical process: wave propagation to and from the RIS, the relative phase offsets of elements with respect to the reference point, and the applied phase configuration of each element. A unified framework for identifying key performance indicators is established by analyzing the conditioning of the sensing operators. In the multi-RIS setting, we derive--via rank analysis--the governing law among the RoI size, the number of elements, and the number of measurements. Upper bounds on the relative error of the least squares reconstruction algorithm are derived. These bounds clarify how key performance indicators affect estimation error and provide valuable guidance for system-level optimization. Numerical experiments confirm that the trend of the relative error is consistent with the theoretical bounds. Finally, we develop a proof-of-concept prototype using universal software radio peripherals and employ a magnitude-only reconstruction algorithm tailored to the system. To the best of our knowledge, this represents the first experimental demonstration of its kind.
\end{abstract}


\section{Introduction} 
\IEEEPARstart{R}econfigurable intelligent surfaces (RISs) have attracted significant attention for their potential to enhance both communication and sensing capabilities in wireless networks \cite{du2024nested,tang2020mimo,wu2019intelligent,basar2019wireless,song2023intelligent}. An RIS consists of cost-effective, well-designed electromagnetic (EM) units, each capable of independently modifying the characteristics of incident EM waves. These units enable the control of the scattering and reflection properties of EM waves through software-based manipulation. While most existing studies have focused on RIS-aided communication for coverage enhancement in both line-of-sight and non-line-of-sight scenarios \cite{WOS,xiong2024optimal,cheng2022reconfigurable}, RIS-enabled wireless sensing has more recently emerged as a rapidly growing research area \cite{zhang2021metalocalization,he2022high}. Among the various sensing tasks, source localization and power estimation are of particular importance, as they serve as the foundation for numerous higher-level applications \cite{romero2022radio,huang2014cooperative,bazerque2009distributed}. 

In current wireless communication systems, channel state information (CSI) constitutes an immediately available resource for sensing tasks~\cite{jiang2019exploiting,ferrand2023wireless,chen2022sensing,armenta2024wireless}. Since CSI is routinely estimated to facilitate data transmission, it can be repurposed for sensing without incurring additional hardware overhead. Extensive research has investigated CSI for applications such as localization~\cite{jiang2019exploiting,11072249}, tracking~\cite{wei2022channel,li2024wifi}, and activity recognition~\cite{liu2024deep,huang2025sensemamba}. However, CSI-based sensing approaches face inherent limitations. Although user-related information is embedded within CSI, it is often encoded in a complex, nonlinear, and high-dimensional manner that strongly depends on the surrounding environment. For example, fingerprint-based localization techniques~\cite{zhao2023nerf2,khatab2021fingerprint} are typically limited to static environments. Moreover, they require substantial labeled data and exhibit poor adaptability to environmental changes.

When RISs are integrated into sensing or communication pipelines, the channel characteristics become markedly more intricate~\cite{li2023riscan,liu2025tris,liu2024risar,li2024star}. The additional propagation paths introduced by RISs further exacerbate the inherently nonlinear and often ill-posed mapping between user-related attributes and the observed CSI. More critically, in most practical deployments, RISs operate in passive or semi-passive modes without active radio-frequency chains, which makes direct channel acquisition virtually impossible and severely undermines the reliability of CSI for sensing purposes~\cite{yuan2021reconfigurable,gomes2023channel}.

The second methodology for wireless sensing moves away from CSI-based approaches, representing a fundamental paradigm shift. This class of methods, which does not rely on CSI, is explicitly designed to achieve specific sensing objectives by directly engineering the underlying sensing mechanism~\cite{he2022high,sun2024computational,rinchi2022compressive,lin2021single}. For localization tasks, for example, these approaches leverage well-established physical models based on wave propagation theory to extract geometry-related information more effectively. By embedding domain knowledge into the sensing process, these physics-guided designs ensure that the localization framework remains physically informed and theoretically rigorous.

Within this methodological framework, two primary categories emerge: forward localization and backward localization. Forward localization draws conceptual inspiration from classical radio detection techniques. In this approach, the sensing system actively illuminates the target with a known probing signal and then processes the reflected or scattered waves to estimate the target's position or other spatial attributes. In the RIS-enabled forward localization paradigm, the RIS is configured to produce highly directive beams that concentrate energy on a designated region of interest (RoI)~\cite{keykhosravi2023leveraging,sun2024computational,he2022high,buzzi2021radar,wymeersch2020radio}. By systematically varying these configurations, the RIS can scan distinct portions of the scene, thereby enabling spatially selective interrogation of the environment.

The primary advantage of this scheme lies in its conceptual clarity and ease of implementation, as it exploits the flexible beamforming capabilities of RISs. However, it also exhibits inherent limitations. In practice, many existing RIS prototypes employ only a one-bit quantization scheme, with each unit cell switching between two phase states~\cite{he2022high}. Such extremely low-resolution quantization not only broadens the beamwidth but also degrades beam-steering accuracy, thereby restricting scanning granularity~\cite{keykhosravi2023leveraging,shekhawat2024millimeter}. Moreover, the grating and side lobes induced by this coarse quantization can introduce ambiguities in the localization process~\cite{vabichevich2023suppression}.

In the backward localization paradigm, the emphasis shifts from generating directive pencil beams to probe the signal distribution toward reconstructing the RoI from a sequence of carefully designed measurements. This approach draws on principles from inverse problem theory~\cite{chen2018computational}, computational imaging~\cite{imani2020review,saigre2022intelligent}, and the well-established compressed sensing framework~\cite{donoho2006compressed}. Within this framework, the RoI is illuminated using a set of diverse radiation patterns that linearly encode spatial information into a series of measurements~\cite{eldar2012compressed}. At the estimation or reconstruction stage, the collect measurements are processed to recover the RoI, including both the positions of sources and their associated power levels~\cite{10096085}. RISs constitute a natural hardware platform for implementing such diverse illumination strategies, as their reconfiguration capabilities enable the synthesis of spatially adaptive radiation patterns across the RoI.

Several notable studies have investigated this direction, showing that backward sensing through RISs can address tasks such as localization and angular information estimation~\cite{lin2021single,alexandropoulos2022localization,rinchi2022compressive,10096085}, computational imaging~\cite{jiang2024near,huang2024ris,li2024radio,hu2022metasketch}, and environmental mapping~\cite{torcolacci2024holographic,tong2021joint}. Early works utilize RISs for direction-of-arrival (DoA) estimation~\cite{lin2021single, alexandropoulos2022localization}, where exploiting sparsity in the angular domain further improves accuracy~\cite{rinchi2022compressive}. Going beyond the far-field assumption, Dardari et al.~\cite{dardari2021nlos} investigate near-field localization with large RISs, demonstrating that the wavefront curvature at high frequencies enables reliable single-anchor positioning even under NLOS conditions. Another line of research employs RISs for EM inverse scattering imaging~\cite{sun2024computational,torcolacci2024holographic,tong2021joint,hu2022metasketch,huang2024ris}. These studies formulate models that relate the scattered fields to object properties and leverage tools such as singular value decomposition \cite{sun2024computational} or optimization-based algorithms to enhance imaging performance \cite{tong2021joint,hu2022metasketch}. In~\cite{huang2024ris}, the authors introduce a user-moving multi-view imaging approach for RIS-aided communication systems. They employ a stacked measurements framework across multiple RIS configurations and adopt the point spread function concept to assess the system's imaging capabilities. Beyond single-RIS configurations, multi-RIS systems provide additional spatial diversity and extend the observable field~\cite{alexandropoulos2022localization,li2024radio,wang2025dreamer,9475155}. Related to this, several works draw inspiration from computed tomography, reconstructing fields through dense angular sampling~\cite{wilson2010radio,karl2023foundations}. Li et al.~\cite{li2024radio} further demonstrate that only a small number of RISs can achieve comparable imaging and sensing performance by exploiting both reconfiguration and spatial diversity, thereby reducing hardware complexity. 

Although the aforementioned studies have shown that RISs can, in principle, facilitate localization and related functionalities, most existing works remain at the proof-of-concept stage, relying primarily on numerical simulations without hardware validation. Furthermore, the theoretical foundations of RIS-centric backward localization remain conspicuously underdeveloped, with existing research offering only fragmented attempts at establishing rigorous analytical underpinnings~\cite{sun2024computational,mi2023towards}. The influence of key design factors has yet to be systematically examined. As a result, current design methodologies are largely heuristic, depending on empirical adjustments rather than being grounded in a robust theoretical framework~\cite{he2022high,9475155,wilson2010radio}.

The difficulty in establishing rigorous theoretical foundations arises from the intrinsic complexity of RIS-mediated wave interactions. 
Although RISs fundamentally operate on the principle of linear superposition, the emergent behaviors are highly nontrivial. The system's overall response results from a complex interplay among tightly coupled factors, including the configuration and topology of RISs, and the geometry of the RoI. Without explicitly expressing the interaction in terms of these factors, quantitative evaluation of key performance metrics-such as resolution, robustness, and scalability-remains infeasible. Currently, the field lacks a unifying theoretical framework capable of guiding system design, optimizing performance trade-offs, and informing practical implementation strategies.

\subsection{Contributions} 
The main contributions are summarized as follows: 

\begin{itemize}
\item \textbf{Multi-RIS backward sensing framework enabled by the inherent structure underlying the measurement procedure:} We develop physically accurate and mathematically concise models to characterize the forward signal aggregation process in RIS-centric localization. Building on our prior work, which was limited to single-snapshot formulations with a single RIS~\cite{mi2023towards}, this study reveals the inherent structure of the sensing operator in the multi-snapshot measurement process involving one or multiple RISs. This structured representation provides a principled framework for systematically analyzing the key performance indicators that determine the performance of backward localization.

\item \textbf{Key indicators dominating the performance of backward sensing:} 
A theoretical framework based on singularity analysis is established to identify the key performance indicators governing backward sensing performance. In the context of source localization and power estimation with multiple RISs, we derive the law governing the relationship among the region of interest (RoI) size, the number of elements, and the number of measurements via rank analysis. A necessary condition for high-fidelity recovery is that the  size of the RoI must be smaller than both the number of elements and the number of measurements. Upper bounds on the relative error of the least squares (LS) estimation algorithm are derived through resolution analysis. These bounds illustrate how key performance indicators govern the error and provide valuable guidance for system-level optimization. Numerical experiments confirm that the trend of the relative error perfectly aligns with that of the theoretical bound. 

\item \textbf{Proof-of-concept prototype validation:} 
Extensive experiments validate the effectiveness of the proposed RIS-centric backward sensing approach. For backward sensing with multiple RISs, numerical trials illustrate the impact of several factors, including including the number of elements, the number of measurements, the number of RISs, and their topology. To demonstrate the practicality of our approach, we develop a proof-of-concept prototype using universal software radio peripherals (USRP) and employ a magnitude-only reconstruction algorithm tailored to the system. To the best of our knowledge, no previous prototypes have been validated under such configurations. 

\end{itemize}

\subsection{Outline}
The remainder of the paper is organized as follows. Section~\ref{S:ForwardAggregation} presents forward signal aggregation models. Section~\ref{S:BackwardSensing} addresses the RIS-centric backward sensing problem. Section~\ref{S:KeyIndicators} discusses key performance indicators, and Section~\ref{S:Experiments} validates our findings through numerical experiments. In Section~\ref{S:POC}, we demonstrate the practicality of our approach with a USRP-based prototype, and Section~\ref{S:Conclusion} concludes the paper.

\subsection{Notations}

Unless explicitly specified, bold capital letters and bold small letters denote matrices and vectors, respectively. The conjugate transpose, transpose, and pseudo-inverse of $\mathbf{A}$ are denoted by~$\mathbf{A}^{*}$,~$\mathbf{A}^\top$ and~$\mathbf{A}^{\dag}$, respectively. We denote by $\| \cdot \|$ the Euclidean norm of a vector and the spectral norm of a matrix. The notation $\text{diag}(\mathbf{a})$ represents a diagonal matrix with diagonal elements $\mathbf{a}$, and rank\{$\mathbf{A}$\} denotes the rank of $\mathbf{A}$. Additionally, the main notations used in this paper and their physical meanings are summarized as Table~\ref{tab:notations}.

\begin{table*}[!htbp] 
  \caption{Physical Meanings of the Main Notations}
    \label{tab:notations}
  \centering
  \setlength{\tabcolsep}{4pt}  
  \begin{tabular}{@{}l c l c@{}}
    \toprule
    \textbf{Physical Meaning} & \textbf{Notation} 
      & \textbf{Physical Meaning} & \textbf{Notation} \\
    \midrule
    Number of RIS units & $N$ & Number of measurements & $T$ \\
    Size of the RoI & $M$ & Number of RIS blocks & $K$ \\
    Incident signal & $\mathbf{E}$ & Observed signal across $T$ measurements & $\mathbf{S}$  \\ 
    Incident angle & $\theta_i$, $\phi_i$ & Reflected angle & $\theta_s$, $\phi_s$\\
    Set of incident angles & $\mathbf{\Theta}^\text{i}$,$\mathbf{\Phi}^\text{i}$  & Set of reflected angles & $\mathbf{\Theta}^\text{s}$ , $\mathbf{\Phi}^\text{s}$\\
    Phase-difference vector & $\mathbf{v} ( \theta^\text{s}, \phi^\text{s} )$ & Phase-difference matrix of linear array & $\mathbf{V} ( \mathbf{\Theta}^\text{i})$ \\
    Phase-difference matrix of planar array & $\mathbf{V} ( \mathbf{\Theta}^\text{i}, \mathbf{\Phi}^\text{i} )$ &Position of the n-th unit ($x$, $y$, $z$ coordinates) & $\mathbf{p}_n = [x_n, \ y_n, \ z_n ]^{\top}$  \\
    Element spacing & $d$ & Unit direction vector in spherical coordinates  & $\mathbf{u} ( \theta, \phi)$\\ 
    Phase configuration of the $n$-th unit & $\Omega_n$ & Unit scattering pattern & $\tau$ \\ 
    Distance between the Tx and the RIS & $r^i$  & Distance between the Rx and the RIS & $r^s$ \\ 
    Path-loss factors of incident and reflected paths & $l(r^i)$, $l(r^s)$ & Signal-to-Noise Ratio & $SNR$ \\
    Noise term & $\nu_t$ & Noise variance & $\sigma^2$  \\
    Frequency of Tx signals & $f$ & Wavelength of Tx signals & $\lambda$  \\
    Sensing operator & $\mathbf{H} ( \mathbf{\Omega} )$ & Angle resolution and cross-range distance & $ \Delta$ and $\Delta_{\text{CR}}$ \\
    \bottomrule 
  \end{tabular}
\end{table*}

\section{Forward Signal Aggregations via A Single RIS}\label{S:ForwardAggregation}
In RIS-centric backward sensing, a deep understanding of the forward measurement process is essential. Mathematically, this process is formulated as $S_t = H_t(\mathbf{E})$, where $S_t$ represents the observation. The vector $\mathbf{E}\in \mathbb{C}^{M \times 1}$ describes the scene under investigation, with $M$ denoting the number of discretized points within the RoI. The sensing operator $H_t$ maps the scene to the observation. We employ $t=1, \ldots, T$ to index measurements corresponding to various configurations. In practice, the random phase configuration mode is widely used as a benchmark for evaluating sensing performance~\cite{mi2023towards,jiang2024near,sleasman2016microwave,dardari2021nlos}.

\subsection{Linear Aggregations of Multiple Incident Signals}
Our analysis begins by concentrating on a single RIS situated in the $xoy$-plane, composed of multiple units located at $\mathbf{p}_n = [x_n, \ y_n, \ z_n ]^{\top}$, $n=1, \ldots, N$. For simplicity, we assume that the scattering pattern of an isolated unit is denoted by $\tau(\theta^\text{s}, \phi^\text{s}; \theta^\text{i}, \phi^\text{i})$. This notation signifies the dependence on both the incident angle $(\theta^\text{i}, \phi^\text{i})$ and the scattered angle $(\theta^\text{s}, \phi^\text{s})$. We do not elaborate on the specific form $\tau(\cdot)$, since our primary objective is to illustrate the inherent structure of the sensing operator that governs the interaction among the array, incident waves, and scattered waves. Using the physical optics method, we calculate the bistatic scattered field of a rectangular metallic patch with near-zero thickness~\cite{mi2023towards, ozdogan2019intelligent}.

Let us consider an example illustrated in Fig.~\ref{F:ArrayGeometry}. The RIS is illuminated by a single incident EM wave originating from $(r^\text{i}, \theta^\text{i}, \phi^\text{i})$. Suppose the incident electric field arriving at $\mathbf{p}_n$ along $(\theta^\text{i} (n), \phi^\text{i} (n))$ is denoted by $E^\text{i} (n) $. Consequently, the scattered filed at $(\theta^\text{s} (n), \phi^\text{s} (n) )$ is represented as $E^\text{s} (n) = \tau ( \theta^\text{s} (n), \phi^\text{s} (n) ; \theta^\text{i} (n), \phi^\text{i} (n) ) E^\text{i} (n)$. As the EM wave propagates toward the unit at $\mathbf{p}_n$, the attenuation behavior is characterized by the factor $r^\text{i} (n) e^{ j 2 \pi r^\text{i} (n) / \lambda }$. Therefore, $E^\text{i} (n) = E ( r^\text{i}, \theta^\text{i}, \phi^\text{i} ) e^{ - j 2 \pi r^\text{i} (n) / \lambda}/ r^\text{i} (n)$. Similarly, if $E^\text{s} (n)$ denotes the scattered electric field of the unit at $\mathbf{p}_n$, then the electric field at the observation point is $S_n = E^\text{s} (n) e^{ - j 2 \pi r^\text{s} (n) / \lambda } / r^\text{s} (n)$. 

\begin{figure}[!htbp]
  \centering
  \includegraphics[width=0.8\linewidth]{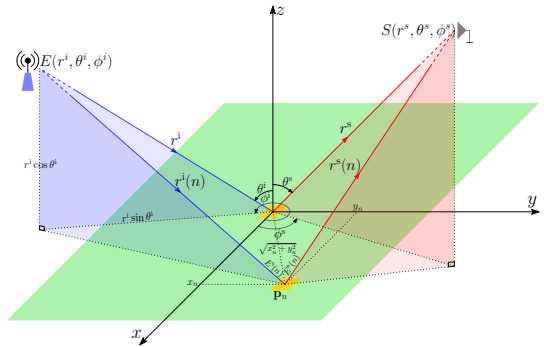}
  \caption{An RIS illuminated by a single incident EM wave from $(r^\text{i}, \theta^\text{i}, \phi^\text{i})$. The scattered waves from various elements are aggregated at $(r^\text{s}, \theta^\text{s}, \phi^\text{s})$.}
  \label{F:ArrayGeometry}
\end{figure}

If the phase configuration of the $n$-th unit is denoted by $\Omega_n$, the corresponding phase delay is expressed as $e^{j \Omega_n}$. In the following, we consider the most common case where the phase assumes random binary values from the set $\{0,\pi \}$. Combining the three components along the propagation path through the $n$-th unit at $\mathbf{p}_n$, the electric field at the observation point is expressed as 
\begin{multline}\label{SISOBehavior}
  S_n ( r^\text{s}_n, \theta^\text{s}, \phi^\text{s} ) = \tau ( \theta^\text{s} (n), \phi^\text{s} (n) ; \theta^\text{i} (n), \phi^\text{i} (n) ) e^{j \Omega_n} \\
  \frac{ e^{ - j 2 \pi r^\text{i} (n) / \lambda} e^{ - j 2 \pi r^\text{s} (n) / \lambda } }{ r^\text{i} (n) r^\text{s} (n) } E ( r^\text{i}, \theta^\text{i}, \phi^\text{i} ) . 
\end{multline}
When the RIS is exposed to multiple incident waves, the aggregated electric field at the observation point is determined by the superposition of the contributions from all units
\begin{equation}\label{Model_MISO}
    \begin{aligned}
      S ( r^\text{s}, \theta^\text{s}, \phi^\text{s} ) = 
      & \sum_{m=1}^{M} E ( r^\text{i}_m, \theta^\text{i}_m, \phi^\text{i}_m )  \\
      & \sum_{n=1}^{N} \tau ( \theta^\text{s} (n), \phi^\text{s} (n) ; \theta^\text{i}_m (n), \phi^\text{i}_m (n) ) \\ 
      & e^{j \Omega_n} \frac{ e^{ - j 2 \pi r^\text{i}_m (n) / \lambda} e^{ - j 2 \pi r^\text{s} (n) / \lambda } }{ r^\text{i}_m (n)  r^\text{s} (n)}.
    \end{aligned} 
\end{equation} 

Under specific scenarios, the input-output expression can be simplified, particularly when the RIS is composed of isotropic units. In isotropic scattering, incident EM waves are uniformly scattered in all directions over the reflection hemisphere, regardless of the angles of incidence or observation. Consequently, the unit's scattering pattern simplifies to  $\tau ( \theta^\text{s}, \phi^\text{s}; \theta^\text{i}, \phi^\text{i} ) = \tau$. The aggregated electric field is given by
\begin{equation}\label{Model_MISO2}
  \begin{aligned}
    S ( r^\text{s}, \theta^\text{s}, \phi^\text{s} ) = 
    & \sum_{m=1}^{M} E ( r^\text{i}_m, \theta^\text{i}_m, \phi^\text{i}_m )  \\
    & \sum_{n=1}^{N} \tau  e^{j \Omega_n} \frac{ e^{ - j 2 \pi r^\text{i}_m (n) / \lambda} e^{ - j 2 \pi r^\text{s} (n) / \lambda } }{ r^\text{i}_m (n)  r^\text{s} (n) }   . 
  \end{aligned}
\end{equation}

{In the far-field regime, the path loss between the units and the source can be decomposed into two components: amplitude attenuation and phase difference. Since the RIS is relatively small compared to the source distance, the amplitude attenuation across different units is nearly constant. Moreover, the phase delays resulting from the path differences can be expressed as steering vectors. Let $\mathbf{u} ( \theta, \phi) = [\sin \theta \cos \phi, \ \sin \theta \sin \phi, \ \cos \theta]^\top$. Then, the representation of the linear aggregation of incident waves is given in \eqref{E:MIMO}.} 

\begin{figure*}[!htbp]
  \begin{equation}\label{E:MIMO}
    \begin{aligned}
        S (r^\text{s}, \theta^\text{s}, \phi^\text{s} ) 
      \approx & 
      \tau 
      \underbrace{ \frac{ e^{-j 2 \pi r^\text{s} / \lambda }}{ r^\text{s} }  }_{\text{Scattered path loss factor } l (r^\text{s})}
      \underbrace{
        \begin{bmatrix}
        e^{ j 2 \pi \mathbf{p}_1^{\top} \mathbf{u} (\theta^\text{s}, \phi^\text{s})  / \lambda } & \cdots & e^{ j 2 \pi \mathbf{p}_N^\top \mathbf{u} (\theta^\text{s}, \phi^\text{s})  / \lambda }  
      \end{bmatrix} }_{ \text{Scattered phase difference vector } \mathbf{v} ( \theta^\text{s}, \phi^\text{s} )}
      \underbrace{
        \begin{bmatrix}
        e^{j \Omega_1} &        &  0 \\
                               & \ddots &    \\
            0                  &        & e^{j \Omega_N}
      \end{bmatrix} }_{\text{Phase configuration } \text{diag} (e^{j \mathbf{\Omega} } ) } \\
      & \underbrace{
        \begin{bmatrix}
        e^{ j 2 \pi \mathbf{p}_1^\top \mathbf{u} (\theta_1^{\text{i}}, \phi_1^\text{i})  / \lambda }  & \cdots  & e^{ j 2 \pi \mathbf{p}_1^\top \mathbf{u} (\theta_M^\text{i}, \phi_M^\text{i})  / \lambda } \\
        \vdots  & \ddots & \vdots \\
        e^{ j 2 \pi \mathbf{p}_N^\top \mathbf{u} (\theta_1^{\text{i}}, \phi_1^\text{i})  / \lambda }  & \cdots  & e^{ j 2 \pi \mathbf{p}_N^\top \mathbf{u} (\theta_M^\text{i}, \phi_M^\text{i})  / \lambda }\\
      \end{bmatrix} }_{\text{Incident phase difference matrix } \mathbf{V} ( \mathbf{\Theta}^\text{i}, \mathbf{\Phi}^\text{i} )}
      \underbrace{
        \begin{bmatrix}
        \frac{ e^{-j 2 \pi r^\text{i}_1 / \lambda } }{ r^\text{i}_1 } &        &  0 \\
                               & \ddots &    \\
            0                  &        &  \frac{ e^{-j 2 \pi r^\text{i}_M / \lambda } }{ r^\text{i}_M }
      \end{bmatrix} }_{\text{Incident path loss factor } \text{diag} (l (\mathbf{r}^\text{i})) }
      \underbrace{
      \begin{bmatrix}
        E (r_1^\text{i}, \theta_{1}^\text{i}, \phi_{1}^\text{i}) \\
        \vdots    \\
        E (r_M^\text{i}, \theta_{M}^\text{i}, \phi_{M}^\text{i})
      \end{bmatrix} }_{\mathbf{E} ( \mathbf{r}^\text{i}, \mathbf{\Theta}^\text{i}, \mathbf{\Phi}^\text{i} )}.
    \end{aligned}
  \end{equation}
  \hrule
\end{figure*}

\subsection{Configurational Diversity}

A significant advantage of the RIS-centric backward sensing approach is its configurational diversity, which enables the generation of a series of measurements (i.e., multiple snapshots), as illustrated in Fig.~\ref{F:SamplingOneAntenna}. We define $\mathbf{E} ( \mathbf{\Theta}^\text{i}, \mathbf{\Phi}^\text{i} ) = \text{diag} ( l (\mathbf{r}^\text{i}) ) \mathbf{E} ( \mathbf{r}^\text{i}, \mathbf{\Theta}^\text{i}, \mathbf{\Phi}^\text{i} )$ as the incident fields arriving at the RIS along $( \mathbf{\Theta}^\text{i}, \mathbf{\Phi}^\text{i} )$. By carefully revisiting \eqref{E:MIMO}, we observe that $\mathbf{v} ( \theta^\text{s}, \phi^\text{s} ) \text{diag} (e^{j \mathbf{\Omega} } ) = e^{j \mathbf{\Omega} } \text{diag} ( \mathbf{v} ( \theta^\text{s}, \phi^\text{s} ) )$. Accordingly, we express the measurement of the aggregated field for the $t$-th snapshot as
\begin{multline}
  S_t (r^\text{s}, \theta^\text{s}, \phi^\text{s} ) = \tau l (r^\text{s}) e^{j \mathbf{\Omega}_t }  \text{diag} ( \mathbf{v} ( \theta^\text{s}, \phi^\text{s} ) ) \\
 \mathbf{V} ( \mathbf{\Theta}^\text{i}, \mathbf{\Phi}^\text{i} ) \mathbf{E} ( \mathbf{\Theta}^\text{i}, \mathbf{\Phi}^\text{i} ) + \nu_t, \ t = 1, \ldots, T . 
\end{multline}

Here, $\nu_t$ represents the noise, which includes modeling inaccuracies due to approximations in the formulation as well as other inherent uncertainties in the measurements. Without loss of generality, we assume the noise follows a distribution with zero mean and variance $\sigma^2$.

\begin{figure}[!htbp] 
  \centerline{\includegraphics[width=0.9\columnwidth]{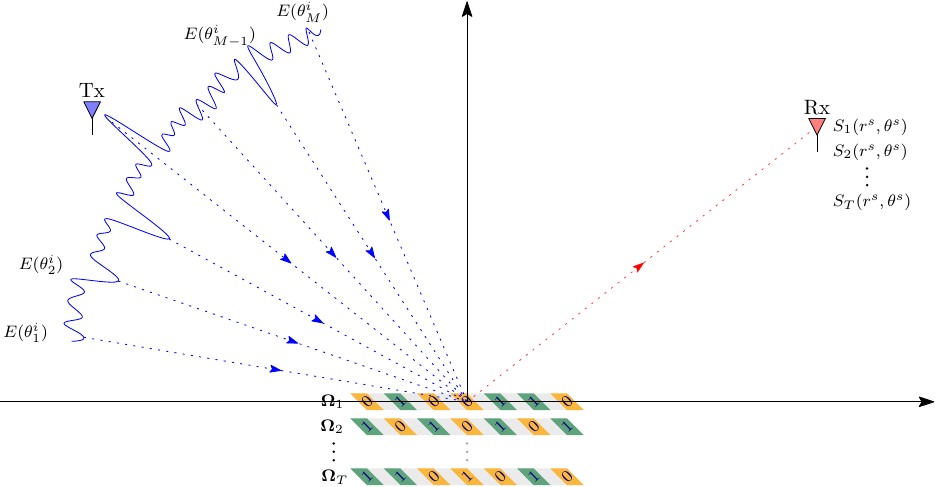}} 
  \caption{Uniform discretization across the azimuth angle domain (RoI). A total of $T$ snapshots are captured corresponding to various configurations.}
  \label{F:SamplingOneAntenna}
\end{figure}

\begin{figure*}[!htbp]
  \begin{equation}\label{E:MultiShot} 
    \begin{aligned}
      \underbrace{
        \begin{bmatrix}
        S_1 (r^\text{s}, \theta^\text{s}, \phi^\text{s}) \\
        \vdots    \\
        S_T (r^\text{s}, \theta^\text{s}, \phi^\text{s})
      \end{bmatrix} }_{\mathbf{S} (r^\text{s}, \theta^\text{s}, \phi^\text{s} ) }
      \approx & 
      \tau \underbrace{ \frac{ e^{-j 2 \pi r^\text{s} / \lambda }}{ r^\text{s} } }_{ \text{Scattered path loss factor } l (r^\text{s}) }
      \underbrace{\begin{bmatrix}
        e^{j \Omega_{1,1}} & \cdots &  e^{j \Omega_{1,N}} \\
        \vdots         & \ddots &  \vdots  \\
        e^{j \Omega_{T,1}} & \cdots &  e^{j \Omega_{T,N}} \\
      \end{bmatrix}}_{\text{Phase configuration matrix }\mathbf{e}^{j \mathbf{\Omega} }} 
      \underbrace{\begin{bmatrix}
        e^{ j 2 \pi \mathbf{p}_1^{\top} \mathbf{u} (\theta^\text{s}, \phi^\text{s})  / \lambda } &        &  0 \\
                               & \ddots &    \\
            0                  &        & e^{ j 2 \pi \mathbf{p}_N^\top \mathbf{u} (\theta^\text{s}, \phi^\text{s})  / \lambda }
      \end{bmatrix}}_{\text{Scattered phase difference matrix } \text{diag} (\mathbf{v} ( \theta^\text{s}, \phi^\text{s} ))} \\
      &
      \underbrace{
      \begin{bmatrix} 
        e^{ j 2 \pi \mathbf{p}_1^\top \mathbf{u} (\theta_1^{\text{i}}, \phi_1^\text{i})  / \lambda }  & \cdots  & e^{ j 2 \pi \mathbf{p}_1^\top \mathbf{u} (\theta_M^\text{i}, \phi_M^\text{i})  / \lambda } \\
        \vdots  & \ddots & \vdots \\
        e^{ j 2 \pi \mathbf{p}_N^\top \mathbf{u} (\theta_1^{\text{i}}, \phi_1^\text{i})  / \lambda }  & \cdots  & e^{ j 2 \pi \mathbf{p}_N^\top \mathbf{u} (\theta_M^\text{i}, \phi_M^\text{i})  / \lambda }\\
      \end{bmatrix} }_{\text{Incident phase difference matrix } \mathbf{V} ( \mathbf{\Theta}^\text{i}, \mathbf{\Phi}^\text{i} )}
      \underbrace{
      \begin{bmatrix}
        \frac{ e^{-j 2 \pi r^\text{i}_1 / \lambda } }{ r^\text{i}_1 } &        &  0 \\
                               & \ddots &    \\
            0                  &        &  \frac{ e^{-j 2 \pi r^\text{i}_M / \lambda } }{ r^\text{i}_M }
      \end{bmatrix} }_{ \text{Incident path loss factor } \text{diag} (l (\mathbf{r}^\text{i})) }
      \underbrace{
      \begin{bmatrix}
        E (r_1^\text{i}, \theta_{1}^\text{i}, \phi_{1}^\text{i}) \\
        \vdots    \\
        E (r_M^\text{i}, \theta_{M}^\text{i}, \phi_{M}^\text{i})
      \end{bmatrix} }_{\mathbf{E} ( \mathbf{r}^\text{i}, \mathbf{\Theta}^\text{i}, \mathbf{\Phi}^\text{i} )}.
    \end{aligned}
  \end{equation}
  \hrule
\end{figure*}

Assuming that the RoI remains stable across snapshots, all $T$ measurements can be arranged into a column vector, yielding the linear representation in \eqref{E:MultiShot}. Let the sensing operator $\mathbf{H} ( \mathbf{\Omega} ) \in \mathbb{C}^{T \times M}$ associated with these snapshots be defined as
\begin{equation}\label{E:H_Omega} 
  \mathbf{H} ( \mathbf{\Omega} ) = \tau l (r^\text{s}) \mathbf{e}^{j \mathbf{\Omega} } \text{diag} \left( \mathbf{v} ( \theta^\text{s}, \phi^\text{s} ) \right) \mathbf{V} ( \mathbf{\Theta}^\text{i}, \mathbf{\Phi}^\text{i} ),
\end{equation}
so that the canonical representation for the measurement procedure is given by
\begin{equation}\label{E:MultiShot2}
  \mathbf{S} (r^\text{s}, \theta^\text{s}, \phi^\text{s} ) = \mathbf{H} ( \mathbf{\Omega} ) \mathbf{E} ( \mathbf{\Theta}^\text{i}, \mathbf{\Phi}^\text{i} ) + \mathbf{n}.
\end{equation}
Here $\mathbf{S} (r^\text{s}, \theta^\text{s}, \phi^\text{s} ) = \left[ S_1 (r^\text{s}, \theta^\text{s}, \phi^\text{s} ), \cdots, S_T (r^\text{s}, \theta^\text{s}, \phi^\text{s} ) \right]^\top$ and $\mathbf{n} = \left[ \nu_1, \cdots, \nu_T \right]^\top$.

{One advantage of the canonical representations in \eqref{E:MIMO} and \eqref{E:MultiShot} is that they provide valuable insight into how RISs interact with both incident and scattered waves. Furthermore, \eqref{E:MultiShot} provides a clear depiction of the underlying structure of the measurement procedure facilitated by RISs. The sensing operator is expressed as the product of a series of matrices, each with a distinct physical interpretation and a limited set of key parameters. This structural representation lays the foundation for systematically analyzing the key indicators influencing performance in backward sensing tasks.

Specifically, the incident and scattered path loss factors $l (\mathbf{r}^\text{i})$ and $l (r^\text{s})$ represent the attenuations and phase shifts along the paths from the sources and receivers to the RIS, respectively. The incident and scattered phase difference factors $\mathbf{V} ( \mathbf{\Theta}^\text{i}, \mathbf{\Phi}^\text{i} )$ and $\mathbf{v} ( \theta^\text{s}, \phi^\text{s} )$ capture the phase differences arising from path variations among different units. For uniformly linear arrays, these factors exhibit well-defined structures. The incident phase difference matrix simplifies to a Vandermonde matrix, as shown
\begin{equation}\label{E:Vandermonde}
  \mathbf{V} ( \mathbf{\Theta}^\text{i})
  = \begin{bmatrix}
    1                                                   & \cdots  & 1                                                   \\
    e^{ j 2 \pi d \sin \theta^{\text{i}}_1 / \lambda }  & \cdots  &  e^{ j 2 \pi d \sin ( \theta^{\text{i}}_M ) / \lambda }             \\
    \vdots                                              & \ddots  & \vdots                                           \\
    e^{ j 2 \pi (N-1) d \sin \theta^{\text{i}}_1 / \lambda } & \cdots & e^{ j 2 \pi (N-1) d \sin ( \theta^{\text{i}}_M ) / \lambda }
\end{bmatrix} .
\end{equation}
This matrix depends solely on the angles and remains independent of the distances. Additionally, the phase configuration matrix $\mathbf{e}^{j \mathbf{\Omega} }$ is well-conditioned when the phase configurations are random binary values, discussed in Section \ref{Spatial Resolution}.

Finally, we note that in \eqref{E:MultiShot2}, $\mathbf{E} ( \mathbf{\Theta}^\text{i}, \mathbf{\Phi}^\text{i} )$ represents the incident fields arriving at the RIS along $( \mathbf{\Theta}^\text{i}, \mathbf{\Phi}^\text{i} )$. In other words, $l(\mathbf{r}^\text{i})$ serves as a distance normalization term in the DoA estimation through a single RIS. This term becomes particularly significant in sensing scenarios involving multiple RISs, where large-scale attenuation and phase shift factors can be detected across different RISs.
}

\section{Backward Sensing Using Multiple RISs} \label{S:BackwardSensing}

Before exploring the capabilities enabled by deploying multiple RISs, we first clarify the DoA estimation function using a single RIS, where the vector $\mathbf{E} ( \mathbf{\Theta}^\text{i}, \mathbf{\Phi}^\text{i} )$ is unknown and must be estimated.

\subsection{DoA Estimation via a Single RIS}

With the model in \eqref{E:MultiShot2}, if the sensing operator $\mathbf{H} ( \mathbf{\Omega} )$ is known, $\mathbf{E} ( \mathbf{\Theta}^\text{i}, \mathbf{\Phi}^\text{i} )$ can be efficiently estimated. According to \eqref{E:H_Omega}, this requires knowledge of $l (r^\text{s})$, $\mathbf{v} ( \theta^\text{s}, \phi^\text{s} )$, $\mathbf{e}^{j \mathbf{\Omega} }$, and $\mathbf{V} ( \mathbf{\Theta}^\text{i}, \mathbf{\Phi}^\text{i} )$. The terms $l (r^\text{s})$ and $\mathbf{v} ( \theta^\text{s}, \phi^\text{s} )$ depend on the location of the receiver, while $\mathbf{e}^{j \mathbf{\Omega} }$ is determined by the configuration. In backward sensing, both terms are readily available.

The situation becomes more complicated for $\mathbf{V} ( \mathbf{\Theta}^\text{i}, \mathbf{\Phi}^\text{i} )$ as this matrix is determined by the directions of the incident waves, which constitutes the central task in DoA estimation. In practice, a commonly used approach involves exploring densely sampled grids across the angular domain, as illustrated in Fig.~\ref{F:SamplingOneAntenna}. While this strategy expands the solution space, the extensive reconfiguration capabilities of RISs ensure that the problem remains well-posed in most scenarios.

{ Without making any prior assumptions about the RoI, when the sensing operator $\mathbf{H} ( \mathbf{\Omega} )$ is known, ensuring that a set of measurements $\mathbf{S} ( r^\text{s}, \theta^\text{s}, \phi^\text{s} )$ uniquely determines a solution for $\mathbf{E} ( \mathbf{\Theta}^\text{i}, \mathbf{\Phi}^\text{i} )$ requires that the rank of the sensing operator $\mathbf{H} ( \mathbf{\Omega} )$ be at least equal to the size of the RoI.} In a typical scenario where the RIS is configured to produce a well-conditioned, full column rank sensing matrix $\mathbf{H} ( \mathbf{\Omega} )$, a reliable estimation of the incident field is given by 
\begin{equation}\label{E:InverseSensing}
  \widehat{\mathbf{E}} ( \mathbf{\Theta}^\text{i}, \mathbf{\Phi}^\text{i} ) = \bigl( \mathbf{H} ( \mathbf{\Omega} )^* \mathbf{H} ( \mathbf{\Omega} ) \bigr)^{-1} \mathbf{H} ( \mathbf{\Omega} )^* \mathbf{S} ( r^\text{s}, \theta^\text{s}, \phi^\text{s} ) .
\end{equation}

\subsection{Cooperative Localization Using Multi-RIS}
Similar to DoA estimation, where incident angles are discretized, cooperative localization with multiple RISs also requires volumetric discretization. A common approach is to perform discretization directly in Cartesian coordinates, typically using uniform voxels, as illustrated in Fig.~\ref{F:TwoModes}. Unlike scenarios involving a single RIS, the operational mode of receivers is essential in sensing with multiple RISs. The simplest mode assigns a dedicated receiver to each RIS, as shown in Fig.~\ref{F:Mode1}. Other modes are also possible; for instance, the reflected signals from all RISs can be aggregated at a single receiver, as illustrated in Fig.~\ref{F:Mode2}. For scenarios involving more than two RISs, a hybrid operational mode can be employed, where some RISs have dedicated receivers while others share a common receiver. Next, we explore the representation of measurement operators in the context of multiple RISs, noting that the form of this operator depends on the chosen operational mode.

\begin{figure}[!htbp]
  \centering
  \subfigure[]{
  \label{F:Mode1}
  \includegraphics[width=.47\columnwidth]{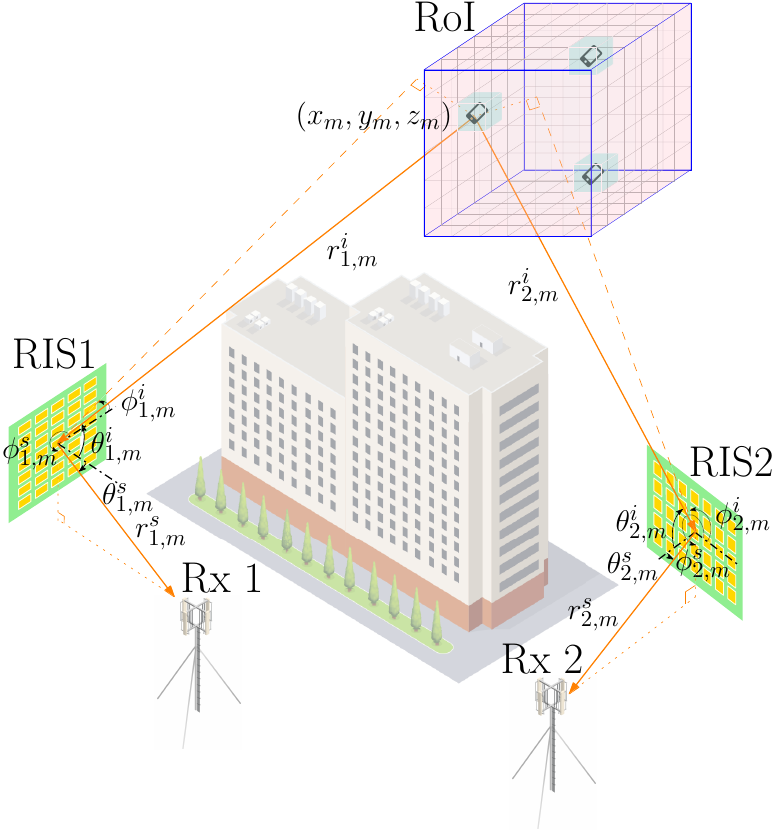}}
  \subfigure[]{
  \label{F:Mode2}
  \includegraphics[width=.47\columnwidth]{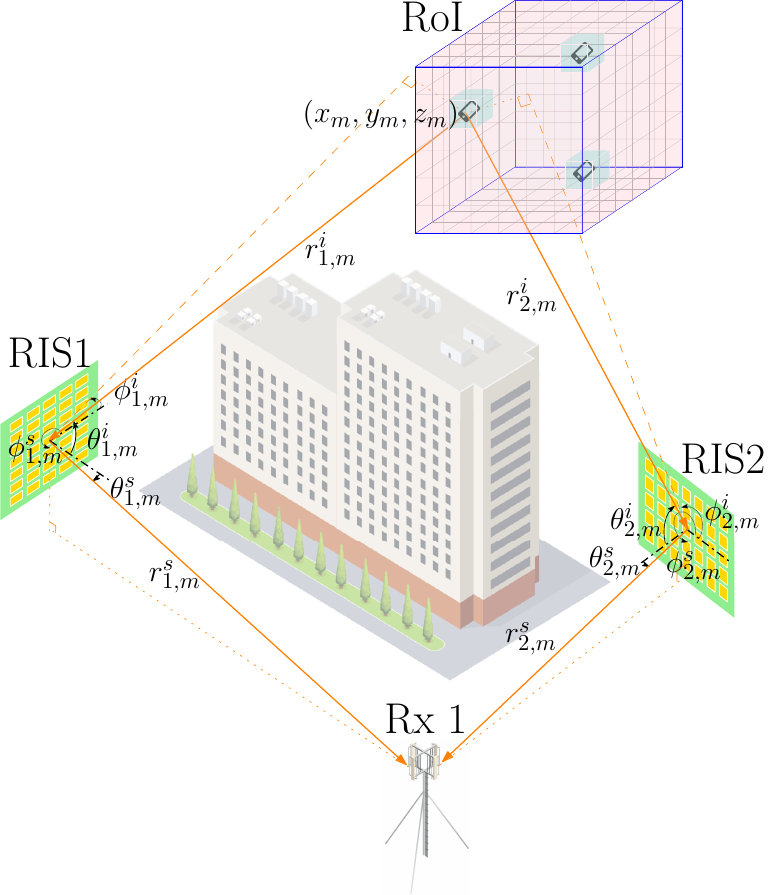}}
  \caption{Two typical operational modes in the RIS-centric sensing system with direct discretization in Cartesian coordinates. (a) Each RIS is equipped with a dedicated receiver, and (b) the reflected signals from all RISs are aggregated at a single receiver.} 
  \label{F:TwoModes} 
\end{figure}

We denote the incident field in Cartesian coordinates as $\mathbf{E} ( \mathbf{x}, \mathbf{y}, \mathbf{z} )$ and use the subscript $k$ to index RISs. The input-output relationship for each individual RIS is given by
\[
  \mathbf{S} (r^\text{s}_k, \theta^\text{s}_k, \phi^\text{s}_k ) = \mathbf{H} ( \mathbf{\Omega}_k ) \text{diag} ( l (\mathbf{r}^\text{i}_k ) ) \mathbf{E} ( \mathbf{x}, \mathbf{y}, \mathbf{z} ) + \mathbf{n}_k .
\]
For the mode where each RIS has a dedicated receiver, the measurement vectors corresponding to each RIS can be stacked into a column vector. Consequently, the measurement process can be expressed as
\begin{equation}\label{E:Mode_1}
  \begin{bmatrix}
    \mathbf{S} (r^\text{s}_1, \theta^\text{s}_1, \phi^\text{s}_1 ) \\
    \vdots \\
    \mathbf{S} (r^\text{s}_K, \theta^\text{s}_K, \phi^\text{s}_K )
  \end{bmatrix}
  =
  \begin{bmatrix}
    \mathbf{H} ( \mathbf{\Omega}_1 ) \text{diag} ( l (\mathbf{r}^\text{i}_1 ) ) \\
    \vdots \\
    \mathbf{H} ( \mathbf{\Omega}_K ) \text{diag} ( l (\mathbf{r}^\text{i}_K ) )
  \end{bmatrix} 
  \mathbf{E} ( \mathbf{x}, \mathbf{y}, \mathbf{z} ) + \mathbf{n}.
\end{equation}
On the other hand, when the sensing system consists of a single receiver serving all RISs, the measurement process can be represented as
\begin{equation}\label{E:Mode_2}
\begin{aligned}
\mathbf{S} ( r^\text{s}, \theta^\text{s}, \phi^\text{s} ) 
= & \mathbf{S} (r^\text{s}_1, \theta^\text{s}_1, \phi^\text{s}_1 ) + \cdots + \mathbf{S} (r^\text{s}_K, \theta^\text{s}_K, \phi^\text{s}_K ) \\
= &  \bigl[ \mathbf{H} ( \mathbf{\Omega}_1 ) \text{diag} ( l (\mathbf{r}^\text{i}_1 ) ) + \cdots \\
& + \mathbf{H} ( \mathbf{\Omega}_K ) \text{diag} ( l (\mathbf{r}^\text{i}_K ) ) \bigr] \mathbf{E} ( \mathbf{x}, \mathbf{y}, \mathbf{z} ) + \mathbf{n}.
\end{aligned}
\end{equation}
When receivers operate in a hybrid mode, where some RISs have dedicated receivers while others share a common receiver, the sensing operator combines elements of both \eqref{E:Mode_1} and \eqref{E:Mode_2}. 

Although the sensing operator varies with the specific geometries, topologies, configurations, and operational modes, the measurement process involving multiple RISs can generally be expressed as
\begin{equation}\label{E:LinearMeasurement}
  \mathbf{S} = \mathbf{H} \mathbf{E} + \mathbf{n} .
\end{equation}
The radio source localization and power reconstruction problem is formulated as follows. Given the observation $\mathbf{S}$ and the sensing operator $\mathbf{H}$, determine an appropriate $\mathbf{E}$ that satisfies \eqref{E:LinearMeasurement}. Similar to DoA estimation using a single RIS, if the length of the observation vector is greater than the degree of RoI, and the sensing operator is of full column rank and well-conditioned, the most straightforward estimation of the RoI is the LS solution to \eqref{E:LinearMeasurement}, which is expressed as
\begin{equation}\label{E:InverseSensingMultiple}
  \hat {\mathbf{E}} = \mathbf{H}^\dag \mathbf{S} = \bigl( \mathbf{H}^* \mathbf{H}  \bigr)^{-1} \mathbf{H}^* \mathbf{S} .
\end{equation}

One advantage of the LS approach is its simplicity and broad applicability. Furthermore, well-established analytical tools are available to assess the performance of LS solutions. While various reconstruction algorithms can yield high-fidelity solutions, exploring them is beyond the scope of this paper.

\section{Key Performance Indicators}\label{S:KeyIndicators}  

We now identify the key indicators that govern the performance of the proposed RIS-centric sensing approaches. An effective sensing operator should efficiently encode a large RoI while remaining robust to observation noise. Mathematically, the operator's rank characterizes its ability to capture the degrees of freedom of the RoI. Thus, a good sensing operator should exhibit a high rank, ensuring that RoIs with diverse patterns are measured effectively. Another critical consideration is noise sensitivity. A robust RoI estimation approach must be inherently insensitive to noise. Numerically, when the sensing operator is of full column rank and well-conditioned, its minimum singular value quantifies the sensitivity of LS estimation to noise. Consequently, we evaluate performance from both the rank and minimum singular value perspectives.

\subsection{Rank of the Sensing Operators}

Suppose the sensing system consists of $K$ RISs, with each RIS comprising $N_k$ elements. These RISs are collaboratively configured to sense an RoI with size $M$. In the operational mode where each RIS has its dedicated receiver, the number of measurements for each RIS is denoted by $T_k$. In the setup where the sensing system has only one receiver, the number of measurements is denoted by $T$. The following theorem clarifies the rank of the sensing operators.

\begin{thm}\label{T:1}
  In the mode where each RIS has its dedicated receiver, the rank of the sensing matrix is bounded above by the size of the RoI, the total number of measurements, and the total number of elements. Specifically, we have
  \[
    \normalfont 
    \text{rank}(\mathbf{H}) \le \min \left \{ M, \sum_{k=1}^K \min \{ T_k, N_k \} \right \}.
  \]
  Furthermore, in the scenario where the sensing system utilizes only one receiver, the rank of the sensing matrix is bounded by
  \[
    \normalfont 
    \text{rank}(\mathbf{H}) \le \min \left \{ M, T,  \sum_{k=1}^K N_k \right \}.
  \]
\end{thm}

\begin{proof}
We first prove the result for the mode where each RIS has its dedicated receiver. Since $\mathbf{H}$ consists of $M$ columns, it directly follows that $\text{rank}(\mathbf{H}) \le M$. { For each $k$, since $\text{diag} ( l (\mathbf{r}^\text{i}_k ) )$ has full rank, we have
\[
\text{rank} \left ( \mathbf{H} ( \mathbf{\Omega}_k ) \text{diag} ( l (\mathbf{r}^\text{i}_k ) ) \right) = \text{rank} \left ( \mathbf{H} ( \mathbf{\Omega}_k )  \right) \le \min \{ T_k, N_k \} .
\]
This inequality arises from the fact that $\mathbf{H} ( \mathbf{\Omega}_k )$ is a product of three matrices, as illustrated in \eqref{E:MultiShot}. The rank of the phase configuration matrix $\mathbf{e}^{j \mathbf{\Omega}_k }$ is bounded by $T_k$, and the rank of the scattered phase difference matrix is $N_k$.} Since $\mathbf{H}$ is the stacking of $K$ such matrices, we obtain 
\[
  \text{rank}(\mathbf{H}) \le \sum_{k=1}^K \min \{ T_k, N_k \} .
\]
A similar proof holds for the mode where the sensing system uses a single receiver.
\end{proof}

In RIS-centric backward sensing schemes, the rank of the sensing matrix limits the number of degrees of freedom that can be measured or encoded from the scene. Theorem~\ref{T:1} asserts that, to achieve high-fidelity recovery of the RoI using the classic LS algorithm, a necessary condition is that the size of the RoI must be smaller than both the total number of elements and the number of measurements. This implies that the size of the RoI is inherently constrained by the physical size of the RISs.

\subsection{The Spatial Resolution Analysis} \label{Spatial Resolution}

We now perform a resolution analysis of LS estimations. The \textit{relative error} of \eqref{E:InverseSensing} and \eqref{E:InverseSensingMultiple} can be expressed as
\begin{equation}\label{E:RelativeError} 
  \begin{aligned}
  \frac{ \lVert \hat {\mathbf{E}} - \mathbf{E} \rVert } { \lVert \mathbf{E} \rVert } 
  = & \frac{ \lVert {\mathbf{H}}^{\dag} \mathbf{H} \mathbf{E} - \mathbf{E} + {\mathbf{H}}^{\dag} \mathbf{n} \rVert } { \lVert \mathbf{E} \rVert } \\
  \le & \frac{ \lVert {\mathbf{H}}^{\dag} \mathbf{H} \mathbf{E} - \mathbf{E} \rVert + \lVert {\mathbf{H}}^{\dag} \mathbf{n} \rVert } { \lVert \mathbf{E} \rVert } \\
  \le & \lVert {\mathbf{H}}^{\dag} \mathbf{H} - \mathbf{I} \rVert + \lVert {\mathbf{H}}^{\dag} \rVert \lVert \mathbf{n} \rVert / \lVert \mathbf{E} \rVert \\
  \approx & \lVert {\mathbf{H}}^{\dag} \mathbf{H} - \mathbf{I} \rVert + \sqrt{T} \lVert {\mathbf{H}}^{\dag} \rVert / \sqrt{\mathrm{SNR}} .
  \end{aligned}
\end{equation}
For simplicity, we define $\mathrm{SNR} = \lVert \mathbf{E} \rVert^2 / \sigma^2$. The relative error is bounded above by the sum of two terms. When the sensing operator $\mathbf{H}$ is of full column rank and well-conditioned, ensuring $\lVert {\mathbf{H}}^{\dag} \mathbf{H} - \mathbf{I} \rVert = 0$ becomes straightforward. Furthermore, if $\lVert {\mathbf{H}}^{\dag} \rVert$ remains bounded by a reasonable value, then recovery using the LS solution guarantees high fidelity. It is worth noting $\lVert {\mathbf{H}}^{\dag} \rVert = 1 / \sigma_{\text{min}} ( \mathbf{H} )$. We now proceed to examine the minimum singular value of the sensing operator $\mathbf{H}$.

In a sensing or imaging system, the minimum resolvable distance is defined as the spatial resolution~\cite{mercuri2017frequency}. This definition stems from the observation that when two point sources are in close proximity, the system cannot distinguish them as separate entities. To simplify the analysis, we focus on investigating the angular resolution associated with backward sensing using a single linear RIS. Consider a scenario where two incident waves impinge on the RIS, one along $\theta^{\text{i}}$ and the other along $\theta^{\text{i}}+\Delta$. In this case, the sensing operator, as derived in \eqref{E:H_Omega}, simplifies to
\[
  \mathbf{H} ( \theta^{\text{i}}, \theta^{\text{i}}+\Delta ) = \tau l (r^\text{s}) \mathbf{e}^{j \mathbf{\Omega} } \text{diag} \left( \mathbf{v} ( \theta^\text{s} ) \right) \mathbf{V} ( \theta^{\text{i}}, \theta^{\text{i}}+\Delta ).
\]

\begin{lem}\label{L:3}
If $\mathbf{B}$ has full column rank, then $\sigma_{\text{min}} ( \mathbf{A} \mathbf{B} ) \ge \sigma_{\text{min}} ( \mathbf{A} ) \sigma_{\text{min}} ( \mathbf{B} )$.
\end{lem}

\begin{proof}
Given that $\mathbf{B}$ has full column rank, for any $x \neq 0$, we have $\mathbf{B} x \neq 0$. Then,
  \begin{multline*}
    \sigma_{\text{min}} ( \mathbf{A} \mathbf{B} ) 
    = \min_{x \neq 0} \frac{ \lVert \mathbf{A} \mathbf{B} x \rVert }{ \lVert x \rVert } 
    = \min_{x \neq 0} \frac{ \lVert \mathbf{A} \mathbf{B} x \rVert }{ \lVert \mathbf{B} x \rVert } \frac{ \lVert \mathbf{B} x \rVert }{ \lVert x \rVert } \\
    \ge \min_{y \neq 0} \frac{ \lVert \mathbf{A} y \rVert }{ \lVert y \rVert } \min_{x \neq 0} \frac{ \lVert \mathbf{B} x \rVert }{ \lVert x \rVert }
    = \sigma_{\text{min}} ( \mathbf{A} ) \sigma_{\text{min}} ( \mathbf{B} ) .
  \end{multline*}
\end{proof}

When $\Delta \neq 0$, the matrix $\mathbf{V} ( \theta^{\text{i}}, \theta^{\text{i}}+\Delta )$ is of full column rank, and the matrix $\text{diag} \left( \mathbf{v} ( \theta^\text{s} ) \right)$ is of full rank. Therefore, we have
\begin{equation}\label{E:SingularValueBound}
  \begin{aligned}
    & \sigma_{\text{min}} ( \mathbf{H} (\theta^{\text{i}}, \theta^{\text{i}}+\Delta) ) \\
    \ge & \lvert \tau l (r^\text{s}) \rvert \sigma_{\text{min}} ( \mathbf{e}^{j \mathbf{\Omega} } )  \sigma_{\text{min}} \left( \text{diag} ( \mathbf{v} ( \theta^\text{s} ) ) \right) \sigma_{\text{min}} \left( \mathbf{V} ( \theta^{\text{i}}, \theta^{\text{i}}+\Delta ) \right) \\ 
    = & \lvert \tau l (r^\text{s}) \rvert \sigma_{\text{min}} ( \mathbf{e}^{j \mathbf{\Omega} } ) \sigma_{\text{min}} \left( \mathbf{V} ( \theta^{\text{i}}, \theta^{\text{i}}+\Delta ) \right) .
  \end{aligned}
\end{equation}

\begin{lem}\label{L:4}
Consider a Vandermonde matrix given by
\begin{equation}\label{E:VandermodeMatrix}
  \mathbf{V} ( \theta^{\mathrm{i}}, \theta^{\mathrm{i}}+\Delta ) =
  \begin{bmatrix}
    1                                                   &  1                                                   \\
    e^{ j 2 \pi d \sin \theta^{\mathrm{i}} / \lambda }             &  e^{ j 2 \pi d \sin ( \theta^{\mathrm{i}}+\Delta ) / \lambda }             \\
    \vdots                                              &  \vdots                                              \\
    e^{ j 2 \pi (N-1) d \sin \theta^{\mathrm{i}} / \lambda } &  e^{ j 2 \pi (N-1) d \sin ( \theta^{\mathrm{i}}+\Delta ) / \lambda }
\end{bmatrix} ,
\end{equation}
where $\Delta \neq 0$. The singular values of $\mathbf{V} ( \theta^{\mathrm{i}}, \theta^{\mathrm{i}}+\Delta )$ are given by
\begin{multline*}
  \sigma_{\text{max}} ( \mathbf{V} ( \theta^{\mathrm{i}}, \theta^{\mathrm{i}}+\Delta ) ) \\
  = \sqrt{ N + \left| \frac{ \sin ( \pi N d (\sin \theta^{\mathrm{i}} - \sin ( \theta^{\mathrm{i}}+\Delta ) ) / \lambda ) }{ \sin ( \pi d ( \sin \theta^{\mathrm{i}} - \sin ( \theta^{\mathrm{i}}+\Delta ) ) / \lambda ) } \right|  } ,
\end{multline*}
and 
\begin{multline}\label{E:LeastSingularValue1}
  \sigma_{\text{min}} ( \mathbf{V} ( \theta^{\mathrm{i}}, \theta^{\mathrm{i}}+\Delta ) ) \\
  = \sqrt{ N - \left| \frac{ \sin ( \pi N d (\sin \theta^{\mathrm{i}} - \sin ( \theta^{\mathrm{i}}+\Delta ) ) / \lambda ) }{ \sin ( \pi d ( \sin \theta^{\mathrm{i}} - \sin ( \theta^{\mathrm{i}}+\Delta ) ) / \lambda ) } \right|  } .
\end{multline}
\end{lem}

\begin{proof}
See the Appendix~\ref{S:Appendix_A3}
\end{proof}

\begin{lem}\label{L:5}
For $x$ sufficiently small, we have
\begin{equation}\label{E:sin_sin}
  \frac{ \sin (N x) }{\sin (x)} \approx N - \frac{ N ( N^2 - 1 ) x^2 }{ 6 } .
\end{equation}
\end{lem}

\begin{proof}
  See the Appendix~\ref{S:Appendix_A4}
\end{proof}

Combining \eqref{E:LeastSingularValue1} and \eqref{E:sin_sin} yields an approximation for the smallest singular value of $\mathbf{V}(\theta^{\text{i}}, \theta^{\text{i}}+\Delta)$
\begin{equation}\label{E:ApproximationV}
\begin{aligned}
  & \sigma_{\text{min}} ( \mathbf{V} ( \theta^{\text{i}}, \theta^{\text{i}}+\Delta ) ) \\
  \approx & \frac{ \pi }{ \sqrt{6} } \frac{ d }{ \lambda } \sqrt{ N ( N^2 - 1 ) } \lvert \sin \theta^{\text{i}} - \sin ( \theta^{\text{i}} + \Delta ) \rvert  \\
  \approx & \frac{ \pi }{ \sqrt{6} } \frac{ d }{ \lambda } \sqrt{ N ( N^2 - 1 ) } \lvert \Delta \rvert \cos \theta^{\text{i}} .
\end{aligned}
\end{equation}
In the last step, we utilize the first-order Taylor expansion of $\sin ( \theta^{\text{i}} + \Delta )$ at $\theta^{\text{i}}$.

Next, we will examine $\sigma_{\text{min}} ( \mathbf{e}^{j \mathbf{\Omega} } )$ using the Marchenko-Pastur law~\cite{marchenko1967distribution}, which characterizes the asymptotic behavior of singular values of large rectangular random matrices. Let's consider a random matrix $\mathbf{X} \in \mathbb{C}^{T \times N}$, whose entries are independent identically distributed random variables with mean 0 and variance 1. We define the empirical spectral distribution of $\mathbf{Y}_N = \mathbf{X}^* \mathbf{X} / T$ as 
$\mu_{ \mathbf{Y}_N } \equiv \frac{1}{N} \sum_{i=1}^{N} \delta_{\lambda_i( \mathbf{Y}_N )}$,
where $\lambda_1( \mathbf{Y}_N ) \le \cdots \le \lambda_N( \mathbf{Y}_N )$ are the eigenvalues (including multiplicity) and $\delta_{ \lambda_i( \mathbf{Y}_N ) }(x)$ is the indicator function $\mathbf{1}_{ \lambda_i( \mathbf{Y}_N ) \le x }$. 

A celebrated theorem in random matrix theory, known as the Marchenko-Pastur distribution, describes the limiting spectral distribution of  $\mathbf{Y}_N$ \cite{marchenko1967distribution, anderson2010introduction}. Under the Kolmogorov condition, i.e., $T, N \to \infty$ and $N / T \to c \in (0, 1)$, $\mu_{ \mathbf{Y}_N } $ converges weakly and almost surely to $\mu_{\text{M-P}}$, with the density $ d \mu_{\text{M-P}} (x) = \frac{1}{2 \pi c x} \sqrt{(x-a)^{+}(b-x)^{+}} d x $. Here, $a = (1 - \sqrt{c})^2$, $b = (1 + \sqrt{c})^2$, and $a^{+} = \max (0, a)$. The Marchenko-Pastur distribution implies, as $T, N \to \infty$, the probability of the eigenvalues of $\mathbf{Y}_N$ lying outside the interval $[a, b]$ approaches zero. 

We now turn our attention to the matrix $\mathbf{e}^{j \mathbf{\Omega}}$, whose entries are independent identically distributed random variables with mean 0 and variance 1. By leveraging the Marchenko-Pastur law, we can obtain a good approximation for its smallest singular value in the regime of large $T$ and $N$,
\begin{equation}\label{E:ApproximationE}
  \sigma_{\text{min}} ( \mathbf{e}^{j \mathbf{\Omega} } ) \approx \sqrt{T a} = \sqrt{T} \left( 1 - \sqrt{ \frac{N}{T}} \right) = \sqrt{T} - \sqrt{N}. 
\end{equation}

Substituting \eqref{E:ApproximationV} and \eqref{E:ApproximationE} into \eqref{E:SingularValueBound} yields
\begin{equation*}
\begin{aligned}
  & \sigma_{\text{min}} ( \mathbf{H} (\theta^{\text{i}}, \theta^{\text{i}}+\Delta) ) \\
  \gtrsim & \frac{ \pi }{ \sqrt{6} } \tau ( r^\text{s} )^{-1} \frac{ d }{ \lambda }  ( \sqrt{T} - \sqrt{N} ) \sqrt{ N ( N^2 - 1 ) } \lvert \Delta \rvert \cos \theta^{\text{i}}  \\
  \approx & \frac{ \pi }{ \sqrt{6} } \tau ( r^\text{s} )^{-1} \frac{ d }{ \lambda }  \left( 1 - \sqrt{ \frac{N}{T} } \right) T^{ \frac{1}{2} } N^{\frac{3}{2}} \lvert \Delta \rvert \cos \theta^{\text{i}}.
\end{aligned}
\end{equation*}
It then follows that
\begin{multline*}
  \left\lVert {\mathbf{H} (\theta^{\text{i}}, \theta^{\text{i}}+\Delta) }^+ \right\rVert \lesssim  \frac{\sqrt{6}}{\pi} r^\text{s} \tau^{-1} \left( \frac{ d }{ \lambda } \right)^{-1}  \left( 1 - \sqrt{ \frac{N}{T} } \right)^{-1} \\
  T^{ - \frac{1}{2} } N^{ - \frac{3}{2} } \lvert \Delta \rvert ^{-1} ( \cos \theta^{\text{i}} )^{-1}. 
\end{multline*}
We finally obtain the main theorem, as detailed below.

\begin{thm}\label{T:RelativeError}
If two closely incident waves along $\theta^{\mathrm{i}}$ and $\theta^{\mathrm{i}}+\Delta$ impinge on a uniform linear RIS, the {\it relative error} for the LS estimation, as given in \eqref{E:InverseSensing}, is approximately bounded by
\begin{multline}\label{E:RelativeErrorBound}
  \frac{ \lVert \hat {\mathbf{E}} - \mathbf{E} \rVert } { \lVert \mathbf{E} \rVert } \lesssim \frac{\sqrt{6}}{\pi} r^{\mathrm{s}}  \tau^{-1} \left( \frac{ d }{ \lambda } \right)^{-1} \left( 1 - \sqrt{ \frac{N}{T} } \right)^{-1} \\
  N^{ - \frac{3}{2} } \lvert \Delta \rvert^{-1} ( \cos \theta^{\mathrm{i}} )^{-1} \mathrm{SNR}^{-1/2}.
\end{multline}
\end{thm}

{ If the cross-range distance between two sources, perpendicular to the azimuth direction, is denoted by $\Delta_{\text{CR}}$, then $\lvert \Delta \lvert \approx \lvert \Delta_\text{CR} \rvert / r^\mathrm{i}$. The following corollary can be derived.

\begin{cor}\label{T:RelativeError2}
If two incident waves, originating from $r^\mathrm{i}$ with a cross-range distance $\Delta_{\text{CR}}$, impinge on a uniform linear RIS, the {\it relative error} for the LS estimation is approximately bounded by
  \begin{multline}\label{E:RelativeErrorBound2}
    \frac{ \lVert \hat {\mathbf{E}} - \mathbf{E} \rVert } { \lVert \mathbf{E} \rVert } \lesssim \frac{\sqrt{6}}{\pi} r^{\mathrm{i}} r^{\mathrm{s}}  \tau^{-1} \left( \frac{ d }{ \lambda } \right)^{-1} \left( 1 - \sqrt{ \frac{N}{T} } \right)^{-1} \\
    N^{ - \frac{3}{2} } \lvert \Delta_{\text{CR}} \rvert^{-1} ( \cos \theta^{\mathrm{i}} )^{-1} \mathrm{SNR}^{-1/2}.
  \end{multline}
\end{cor}
}

\subsection{Scaling Behavior Analysis} 
Theorem~\ref{T:RelativeError} and Corollary~\ref{T:RelativeError2} elucidate the scaling law governing the relative error in backward sensing using a single RIS. As shown in \eqref{E:RelativeErrorBound}, three key factors, aside from the SNR, influence the relative error. The first category relates to the topology and geometry of the RIS, including parameters such as $\tau$ (the scattering pattern of the element), $N$ (the number of elements) and $d$ (the spacing between elements). Rewriting $( d / \lambda )^{-1} N^{ - \frac{3}{2} }$ as $( N d / \lambda )^{-1} N^{ - \frac{1}{2} }$ reveals that the bound scales inversely with $N d$, which represents the RIS aperture. Unlike forward beamforming, where undesirable grating lobes arise when $d / \lambda > 1/2$, there is no such constraint on the element spacing $d$ in backward sensing. In fact, increasing $d$ can enhance the aperture and, consequently, improve the performance. However, increasing $N$ (the number of elements) may yield even greater benefits. This is because the power reflected by RISs grows quadratically with the total collecting area of the elements.

The second category influencing sensing performance is related to the incident angles and the receiver's position. As indicated in \eqref{E:RelativeErrorBound}, the recovery error is directly proportional to $r^\text{s}$. This is because, as the receiver approaches the RIS, it captures more power, leading to a reduction in recovery error. The term $( \cos \theta^{\mathrm{i}} )^{-1}$ indicates that recovery performance degrades as the scanned range deviates from the RIS's normal direction. The terms $\lvert \Delta \rvert^{-1}$ and $\lvert \Delta_{\text{CR}} \rvert^{-1}$ provide insight into the spatial resolution. Specifically, for small values of $\lvert \Delta \rvert$, the matrix $\mathbf{V} ( \theta^{\text{i}}, \theta^{\text{i}}+\Delta )$ becomes ill-conditioned, which in turn causes the sensing operator $\mathbf{H} ( \theta^{\text{i}}, \theta^{\text{i}}+\Delta )$ to be ill-conditioned as well. In such cases, even small amounts of measurement noise can result in substantially larger recovery errors.

The third category involves the number of measurements. The scaling factor between the bound and the number of measurements $T$ is $( 1 - \sqrt{ N / T } )^{-1}$, which shows that the bound is inversely proportional to $1 - \sqrt{ N / T }$. When $T \gg N$, this term asymptotically approaches 1, its minimum value. As a general guideline, increasing the number of measurements significantly beyond the number of elements can improve the performance of backward sensing.

\subsection{Numerical Validation of the Relative Error Bound} 

We now validate the accuracy of the bounds presented in Theorem~\ref{T:RelativeError} and Corollary~\ref{T:RelativeError2}. To this end, we construct a scenario in which only two incident waves impinge on the RIS, and compare the relative error with the theoretical bound to illustrate the influence of key factors, including $N$, $d$, $T$, $\Delta_\text{CR}$, $\theta^\text{i}$, and SNR.

In the following tests, the RIS operates at a frequency of $f = 5.8 \, \text{GHz}$ and comprises unit cells with an effective area of $0.4\lambda \times 0.4\lambda$ (resulting in $\tau=0.16\lambda$), employing random phase configurations. The distance from the center of the RoI to the RIS is $r^\text{i} = 6 \, \text{m}$. The default values for the other parameters are $N = 160$, $T = 500$, $d = 2.6 \, \text{cm} \approx 0.5 \lambda$, $\theta^\text{i} = 0$, $\Delta_\text{CR} = 2 \, \text{cm}$,  $r_s = 1 \, \text{m}$, and SNR = 2000. In these tests, all parameters are held constant except for one, which is systematically varied to evaluate its effect. The relative error of LS estimation is computed as an average over 500 trials.

Fig.~\ref{all_resolution} presents the actual relative error (in blue) alongside the upper bound given by \eqref{E:RelativeErrorBound2} (in red). The plots clearly indicate that the trend of the actual relative error is in perfect agreement with that of the theoretical bound. In other words, the theoretical curve effectively captures the law governing the relationship between the error and the key performance indicators. These theoretical bounds, as given by \eqref{E:RelativeErrorBound} and \eqref{E:RelativeErrorBound2}, provide valuable guidance for system-level optimization. We also observe a gap between the relative error and the theoretical curve, which can be attributed to the relatively loose estimation of the minimum singular value of the sensing operator. If a tighter estimation were obtained, this gap could be reduced.

\begin{figure*}[!htbp] 
  \centerline{\includegraphics[width=2\columnwidth]{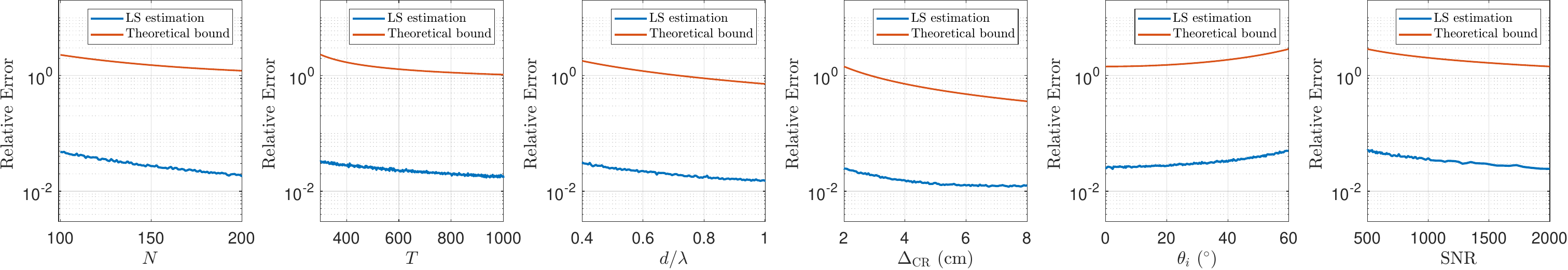}} 
  \caption{Comparison of simulation relative error and theoretical upper bound with respect to the number of elements \(N\), the number of measurements \(T\), the element spacing \(d\), the resolution of cross-range distance $\Delta_\text{CR}$, the incident angle $\cos(\theta_i)$ and the SNR, respectively.} 
  \label{all_resolution} 
\end{figure*}

\section{Performance Evaluation of Backward Sensing Using Multiple RISs} \label{S:Experiments}
In this section, we focus on validating the performance of backward sensing using multiple RISs. Specifically, our goal is to sense a region in which multiple power sources of various shapes are manually positioned. To assess the quality of backward sensing, we employ two metrics: relative error and the structural similarity index measure (SSIM).

The SSIM is a perceptual metric to quantify the degradation of the recovered field $\hat {\mathbf{E}}$ compared to the ground truth $\mathbf{E}$, defined as
\[
  \text{SSIM}=\frac{ \left( 2 \mu_{ \mathbf{E} } \mu_{ \hat {\mathbf{E}} } + c_1 \right) \left( 2 \sigma_{ \mathbf{E} \hat {\mathbf{E}}} + c_2 \right) } { \left( \mu_{\mathbf{E}}^{2}+\mu_{ \hat {\mathbf{E}} }^{2}+c_1\right) \left( \sigma_{\mathbf{E}}^{2}+\sigma_{ \hat {\mathbf{E}} }^{2}+c_2 \right) } .
\]
Here $\mu_{ \hat {\mathbf{E}} }$ and $\mu_{ \mathbf{E} }$ are the means of $\hat {\mathbf{E}}$ and $\mathbf{E}$, $\sigma^2_{\hat {\mathbf{E}}}$ and $\sigma^2_{\mathbf{E}}$ are the variances, and $\sigma_{ \mathbf{E} \hat {\mathbf{E}}}$ is the covariance between $\mathbf{E}$ and $\hat {\mathbf{E}}$. The constants $c_1$ and $c_2$ are used to stabilize the division with weak denominator. The SSIM value lies within the range of $[0, 1]$, where a higher value indicates that the recovered field exhibits structures more similar to the ground truth. 

\begin{figure}[!htbp] 
  \centerline{\includegraphics[width=1\columnwidth]{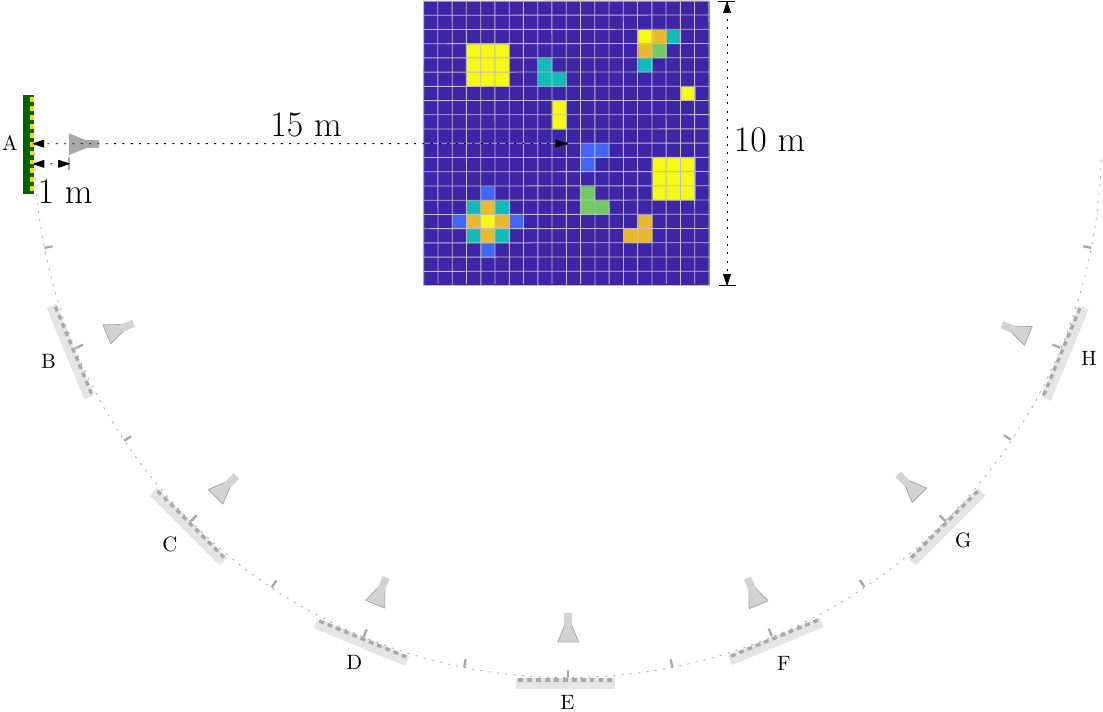}}
  \caption{Scenario configuration used in the simulation. A 10-meter by 10-meter square region is discretized into a 20 $\times$ 20 grid with 0.5-meter intervals. }
  \label{PK_scene}
\end{figure}

\subsection{Impact of the Number of Elements and Measurements} 
Theorem~\ref{T:1} asserts that the rank of the sensing matrix is upper-bounded by the total number of measurements and the total number of elements. To assess the influence of these factors, we first examine the performance of backward sensing in the noiseless case.

Let's imagine a scenario where the task is to sense a $10$ meters by $10$ meters square region with a radio source power distribution. The topology are demonstrated in Fig.~\ref{PK_scene}. Eight landmark positions are set 15 meters around the RoI, labeled from A to H. These landmarks are spaced at 22.5$^\circ$ intervals relative to the RoI, serving as candidate locations for RIS deployment. In this subsection, four uniform linear RISs are placed at four landmark positions, with the configuration \{A, C, E, G\}. Each RIS is equipped with a dedicated receiver. The measurement vectors corresponding to each RIS are stacked into a column vector, as depicted in \eqref{E:Mode_1}. The system operates at a frequency of $f = 20$ GHz. The RoI is discretized into 400 pixels. The parameters are detailed in Table \ref{tab:simulation_parameters}. 

\begin{table}[!t] 
  \centering
  \caption{Default setup parameters used in the simulation experiments.} 
  \label{tab:simulation_parameters} 
  \begin{tabular}{lll} 
  \toprule 
  \textbf{Parameter}                   & \textbf{Symbol}                                    &  \textbf{Value}     \\ \midrule 
  RoI                                  &   $x \times y $                                    &  10 m $\times$ 10 m   \\
  Spatial resolution                   &  $ \bigtriangleup  x \times \bigtriangleup y $ &   0.5 m $\times$ 0.5 m    \\
  Number of pixels                     &  $ M_x \times M_y$                                 &   20 $\times$ 20    \\ 
  Number of RISs                       &  $K$                                               &   4                 \\
  Operating frequency                  &  $f$                                               &   $20$ GHz          \\
  RIS element distance                 &  $d$                                               &   $0.015$ m          \\
  Distance from RoI to RIS             &  $L$                                               &   15 m               \\
  \bottomrule 
  \end{tabular} 
\end{table} 

\begin{figure}[t] 
  \centerline{\includegraphics[width=1\columnwidth]{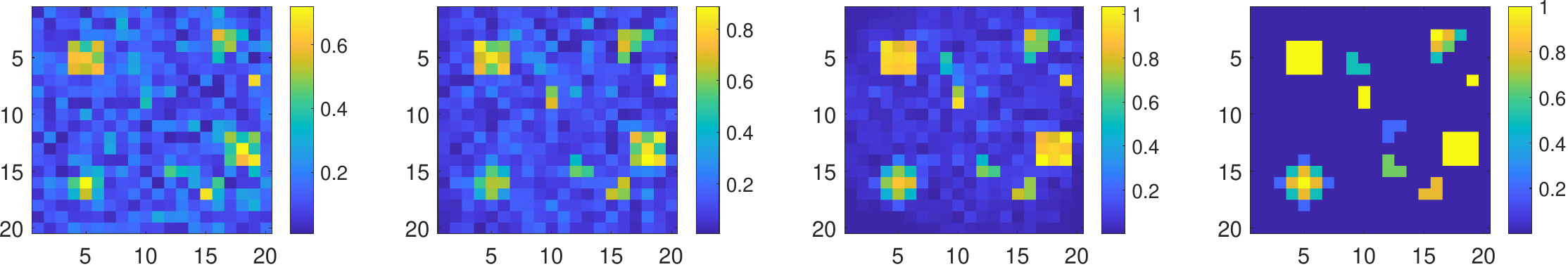}}
  \caption{Comparison of the recovered results as the number of measurements $T_k$ for each RIS varied at 50, 70, 90, and 110, with the number of elements fixed at $N_k = 110$ for each $k$.}
  \label{T_ls_matrix}
\end{figure}

\begin{figure}[t] 
  \centerline{\includegraphics[width=1\columnwidth]{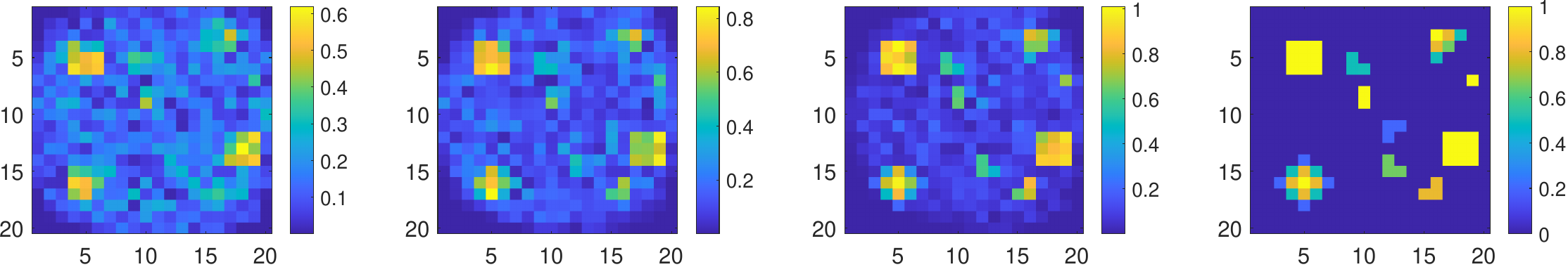}}
  \caption{Comparison of the recovered results as the number of elements $N_k$ for each RIS varied at 50, 70, 90, and 110, with the number of measurements fixed at $T_k = 110$ for each $k$.}
  \label{N_ls_matrix}
\end{figure}

\begin{figure}[!htbp]
  \centering
  \subfigure[]{
  \label{SSIMofT}
  \includegraphics[width=.48\columnwidth]{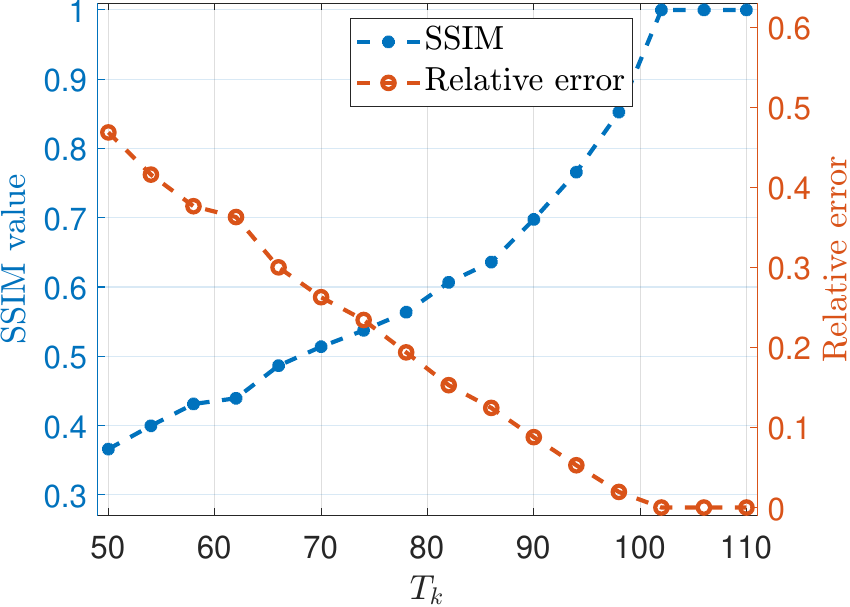}}
  \subfigure[]{
  \label{SSIMofN}
  \includegraphics[width=.48\columnwidth]{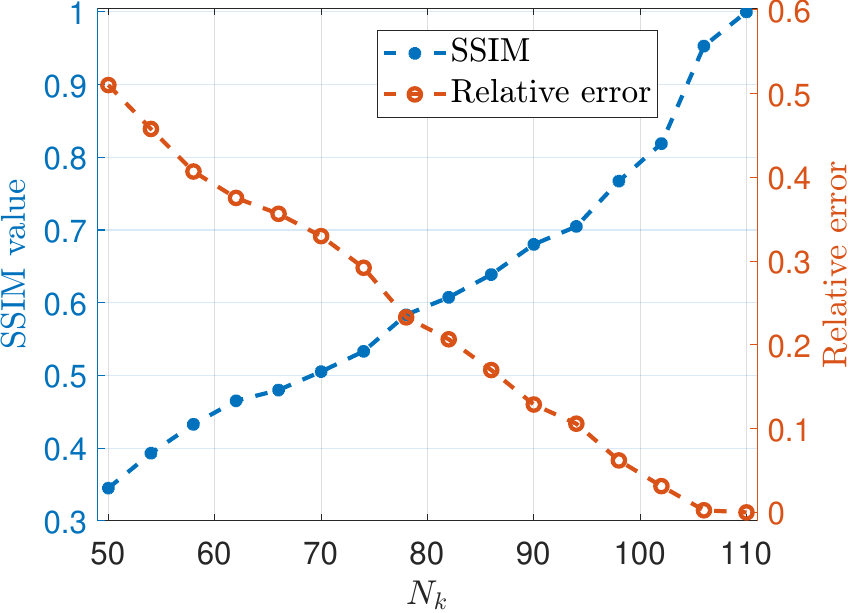}}
  \caption{(a) Relative error and SSIM as the number of measurements $T_k$ for each RIS varied from 50 to 110, with the number of elements fixed at $N_k = 110$ for each $k$. (b) Relative error and SSIM as the number of elements $N_k$ for each RIS varied from 50 to 110, with the number of measurements fixed at $T_k = 110$ for each $k$.}
  \label{SSIMofTN} 
\end{figure}

\begin{figure}[!htbp]
  \centering
  \subfigure[]{
  \label{SVD_1}
  \includegraphics[width=.47\columnwidth]{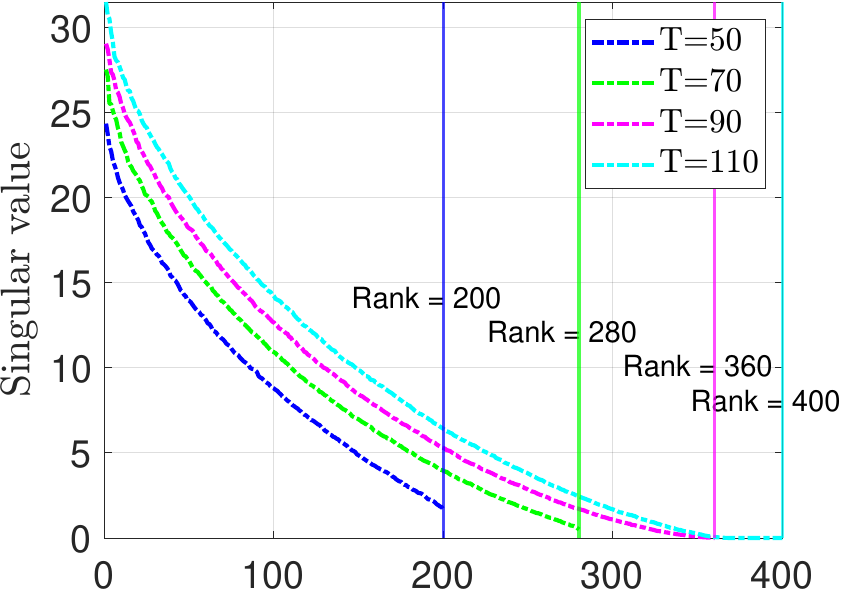}}
  \subfigure[]{
  \label{SVD_2}
  \includegraphics[width=.47\columnwidth]{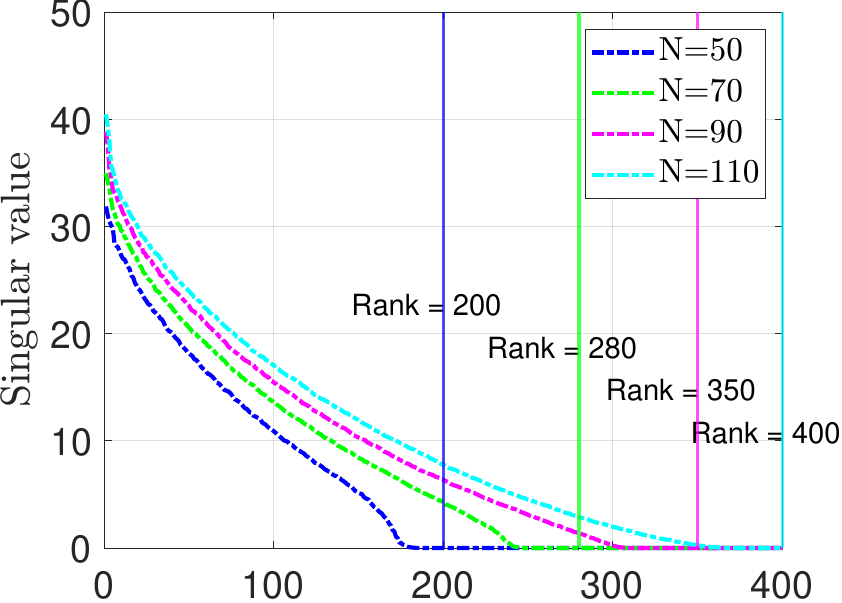}}
  \caption{Sorted singular values corresponding to all trials. (a) The number of measurements $T_k$ is varied at 50, 70, 90, and 110, with the number of elements fixed at $N_k = 110$. (b) $K$ The number of elements $N_k$ is varied at 50, 70, 90, and 110, with the number of measurements fixed at $T_k = 110$.}
  \label{SVD_Illustration} 
\end{figure}

For a comprehensive illustration, we systematically record the recovered results as the number of measurements $T_k$ for each RIS varies from 50 to 110, with the number of elements fixed at $N_k = 110$ for each $k$. The total number of elements sums up to $N = \sum_{k=1}^4 N_k = 440$. The reconstructed scenes are illustrated in Fig.~\ref{T_ls_matrix}. We plot the relative error and SSIM values in Fig.~\ref{SSIMofT}. Additionally, we increase the number of elements of each RIS from 50 to 110, while keeping the number of measurements fixed at $T_k = 110$ (yielding a total of $T = \sum_{k=1}^4 T_k = 440$ measurements). The results are presented in Figs.~\ref{N_ls_matrix} and \ref{SSIMofN}.

These plots demonstrate that increasing both the number of measurements and the number of elements enhances the performance of backward sensing. When $N$ or $T$ is not sufficiently large, a fundamental limitation arises: the sensing matrix encodes the high-dimensional scene into a low-dimensional observation. Clearly, this dimensionality reduction makes recovering the original scene using linear operators impractical. For instance, with $T_k = 50$, accurately reconstructing an original vector in $\mathbb{C}^{400}$ from an observation lying in $\mathbb{C}^{200}$ becomes challenging. { For better understanding, we illustrate the singular values corresponding to all the aforementioned trials in Fig.~\ref{SVD_Illustration}. The rank of the sensing matrix is also computed using singular value decomposition with a tolerance of \(10^{-12}\). As \(T_k\) increases, the rank progressively increases to 200, 280, 360, and 400, respectively. Similarly, for varying values of \(N_k\), the rank increases to 200, 280, 350, and 400, respectively.} These simulation results clearly validate the conclusion of Theorem~\ref{T:1}, confirming that rank deficiency occurs when the number of RIS elements or measurements falls below the RoI size.

It's worth emphasizing that even if the sensing matrix encodes the scene vector into a high-dimensional observation, achieving high-precision recovery might still be challenging. The presence of noise, particularly in cases of ill-conditioned sensing matrices, makes high-fidelity reconstruction from observations unfeasible.

\subsection{Impact of Geometry and Topology} 
In the context of backward sensing using multiple RISs, the conditioning of the sensing matrix depends not only on the number of elements and measurements but also on factors such as the system's geometry and topology. Optimizing the topology of RISs has the potential to enhance the conditioning of the sensing matrix, thereby improving the robustness of source distribution reconstruction.

To demonstrate this improvement, we conduct a series of experiments by strategically selecting eight landmark positions around the RoI, each of which serves as a potential location for an RIS. These RISs are positioned tangentially to a circle encompassing the RoI, which is located in the far field of the RISs. The location of these landmarks are shown in Fig.~\ref{PK_scene}. The specific positions are detailed in Table~\ref{tab:KP_parameters}. To ensure a fair comparison, we maintain the total number of elements at $N = \sum_{k} N_k = 540$ and the total number of measurements at $T = \sum_{k} T_k = 540$. To assess the stability of the sensing system, additive Gaussian noise is introduced in these experiments, with the SNR set to 30.

\begin{table}[!htbp]
  \centering
  \caption{Deployment strategies for combinations of different candidate positions.} 
  \label{tab:KP_parameters} 
  \begin{tabular}{lll} 
  \toprule 
  \textbf{Category}                        & \textbf{Configuration}              &  \textbf{Description}                                            \\ \midrule 
  Strategy \uppercase\expandafter{\romannumeral1}:     &  \{A\}                           &  \makecell[l]{-Use only a single landmark position }                 \\
  Strategy \uppercase\expandafter{\romannumeral2}:    &  \{A, E\}                        &  \makecell[l]{-Use two landmark positions perpendi-\\cular to the RoI}  \\
  Strategy \uppercase\expandafter{\romannumeral3}:  &  \{A, C, E, G\}                  &  \makecell[l]{-Use four landmark positions with adj-\\acent interval at an angle of 45$^\circ$} \\ 
  Strategy \uppercase\expandafter{\romannumeral4}:  & \makecell[l]{\{A, B, C, D,\\ E, F, G, H\} } &  \makecell[l]{-Use all landmark positions with high \\landmark density}      \\  
  \bottomrule 
  \end{tabular} 
\end{table}

The results of the backward sensing experiments are presented in Fig.~\ref{KP_ls_matrix}. In addition to the reconstructed source distribution, we plot the sorted singular values of the sensing matrices in Fig.~\ref{SVSofsensingmatrixKP}. For better illustration, the y-axis uses a logarithmic scale. It is evident that recovering the scene with acceptable fidelity is impossible when using only 1 or 2 RISs in this specific scenario. Notably, when using 2 RISs, the condition number of the sensing matrix is greater than $2.9 \times 10^{5}$. The backward sensing procedure is highly sensitive to perturbations in the observations. { Numerical results demonstrate that the multi-view observations enabled by multi-RIS configurations provide significant sensing gains~\cite{huang2024ris}.} 

{ For a more comprehensive analysis, we further explore how the resolution is affected by the distance between the RISs and the RoI. Using the same number of RISs, we vary the distance between the RISs and the RoI from 10 m to 15 m, 20 m, and 25 m. Fig.~\ref{L_ls_matrix} illustrates the reconstructed results for these varying distances. Closer observation positions yield clearer maps, primarily because shorter distances brings RISs with a larger field of view. Fig.~\ref{SVSofsensingmatrixL} also shows that as the distance increases, the conditioning of the sensing matrix worsens.} 

\begin{figure}[t] 
  \centerline{\includegraphics[width=1\columnwidth]{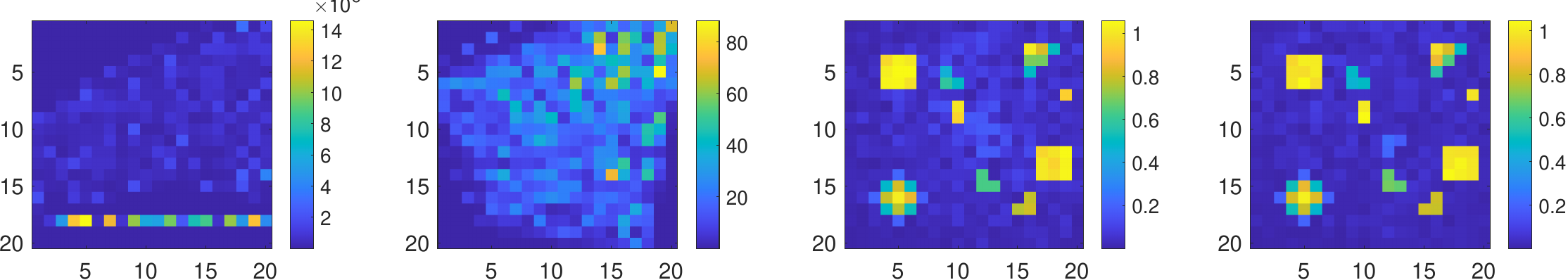}} 
  \caption{Comparison of the recovered results using various geometries and topologies of RISs. Achieving acceptable fidelity in scene recovery is impossible when employing only 1 or 2 RISs in this specific scenario.}
  \label{KP_ls_matrix} 
\end{figure}

\begin{figure}[t] 
  \centerline{\includegraphics[width=1\columnwidth]{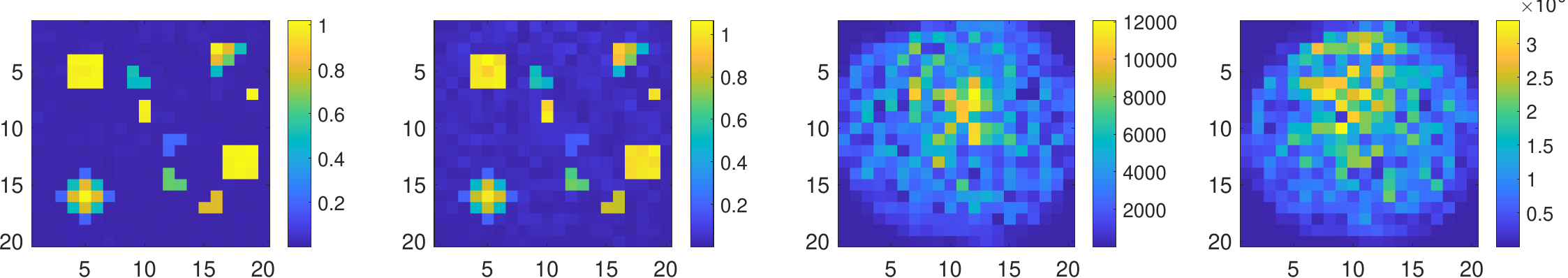}} 
  \caption{Comparison of the recovered results with various distance between the RISs and the RoI from 10 m to 15 m, 20 m, 25 m.}
  \label{L_ls_matrix} 
\end{figure}

\begin{figure}[t] 
  \centering
  \subfigure[]{
  \label{SVSofsensingmatrixKP}
  \includegraphics[width=.47\columnwidth]{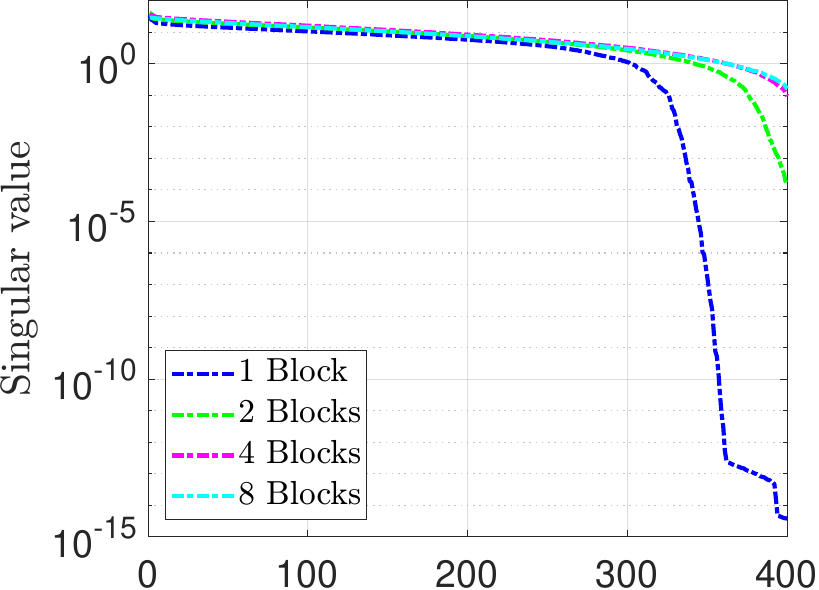}}
  \subfigure[]{
  \label{SVSofsensingmatrixL}
  \includegraphics[width=.465\columnwidth]{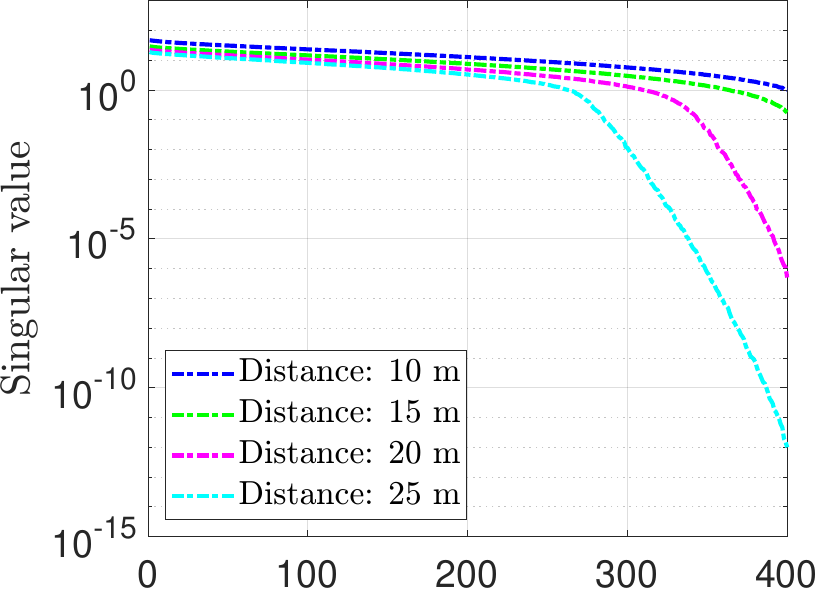}}
  \caption{(a) Sorted singular values for each deployment strategy, with the y-axis presented on a logarithmic scale for clarity. (b) Sorted singular values for varying RIS-to-RoI distance, with the y-axis presented on a logarithmic scale for clarity.} 
  \label{KP_cond_number_SVD}
\end{figure}

{ 
Overall, in multi-RIS deployments, increasing the number of observation views and reducing the distance between the SoI and the RISs to expand the field of view can significantly enhance sensing performance and improve mapping accuracy. }

\section{Proof-of-concept Prototype} \label{S:POC}
To validate the proposed sensing scheme, we implement a proof-of-concept prototype using a universal software radio peripheral (USRP) in a microwave anechoic chamber. Two RISs are strategically positioned at one end of the chamber to detect a power source located at the other end, as illustrated in Fig.~\ref{fig:experiment_scenario}. The schematic is depicted as Fig.~\ref{fig:two_RISs Schematic}. The RISs, operating at a frequency of 5.8 GHz, consist of 10 $\times$ 16 1-bit elements with an element spacing of 0.025 m, as shown in Fig.~\ref{fig:single board}. To evaluate the performance of uniform linear RISs within a planar RoI, each column is configured identically. The two RISs, separated by 2.81 m, are employed to detect the source positioned 6 m away. For each RIS, $T=500$ measurements are acquired to construct the sensing matrix $\mathbf{H}$.

\begin{figure}[tbp]
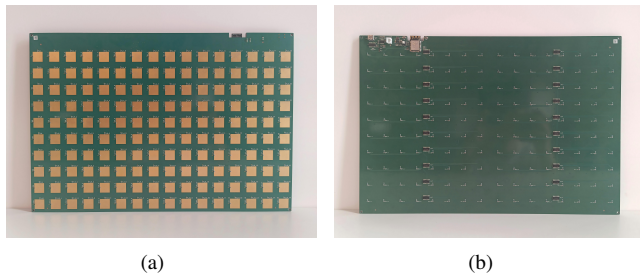
 
  \centering
  \subfigure[]{
  \label{fig:single board1}
  \includegraphics[width=.465\columnwidth]{./figure/3-1-RIS_board.pdf}}
  \subfigure[]{
  \label{fig:single board2}
  \includegraphics[width=.465\columnwidth]{./figure/3-1-RIS_board_back.pdf}}
  \caption{Photographs of the front (a) and back (b) of the fabricated RIS prototype with 10$\times$ 16 elements. The RIS, operating at a frequency of 5.8 GHz, consist of 10 $\times$ 16 1-bit elements with an element spacing of 0.025 m.}
  \label{fig:single board}
\end{figure}

\begin{figure}
  \centerline{\includegraphics[width=1\columnwidth]{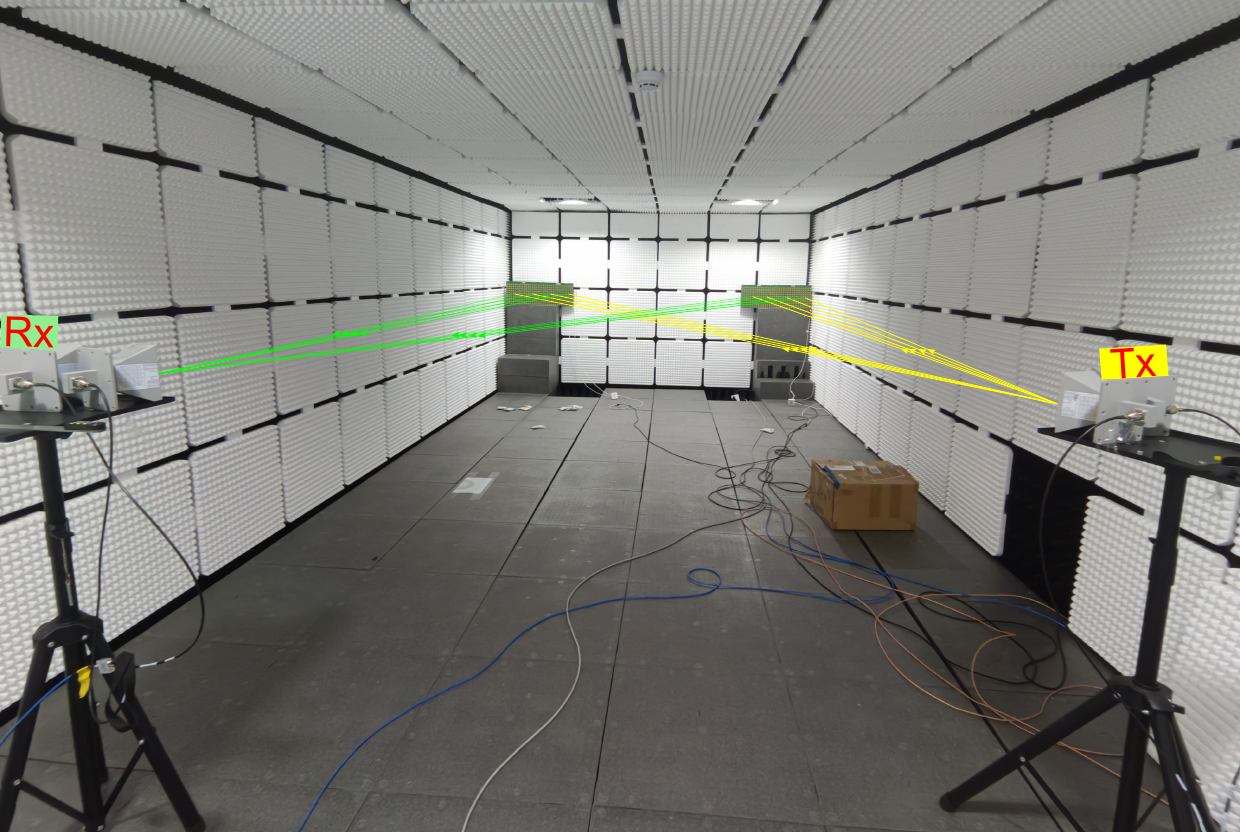}}
  \caption{Experimental setup of two RISs sensing system in a microwave anechoic chamber.}
  \label{fig:experiment_scenario}
\end{figure}

\begin{figure}
  \centerline{\includegraphics[width=0.9 \columnwidth]{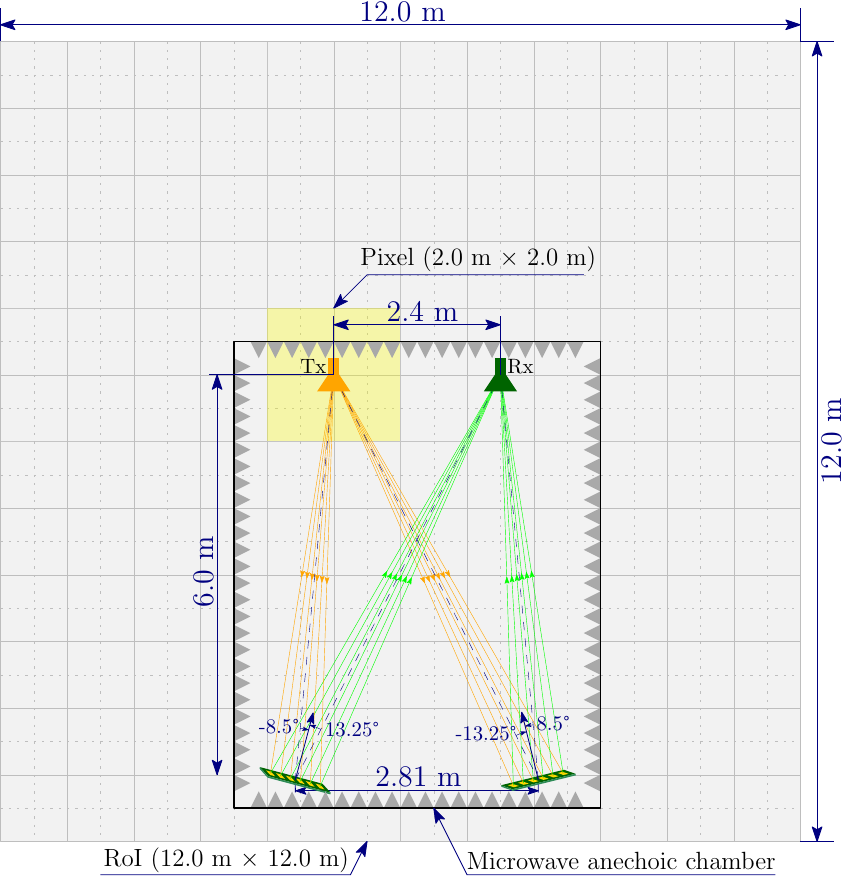}}
  \caption{The geometry and topology of the experiment conducted inside the anechoic chamber. The RoI is a square region of 12 m $\times$ 12 m, discretized into a 12 $\times$ 12 grid, with each pixel covering 1 m $\times$ 1 m.}
  \label{fig:two_RISs Schematic}
\end{figure}

\begin{figure}[!htbp] 
  \centerline{\includegraphics[width=0.9 \columnwidth]{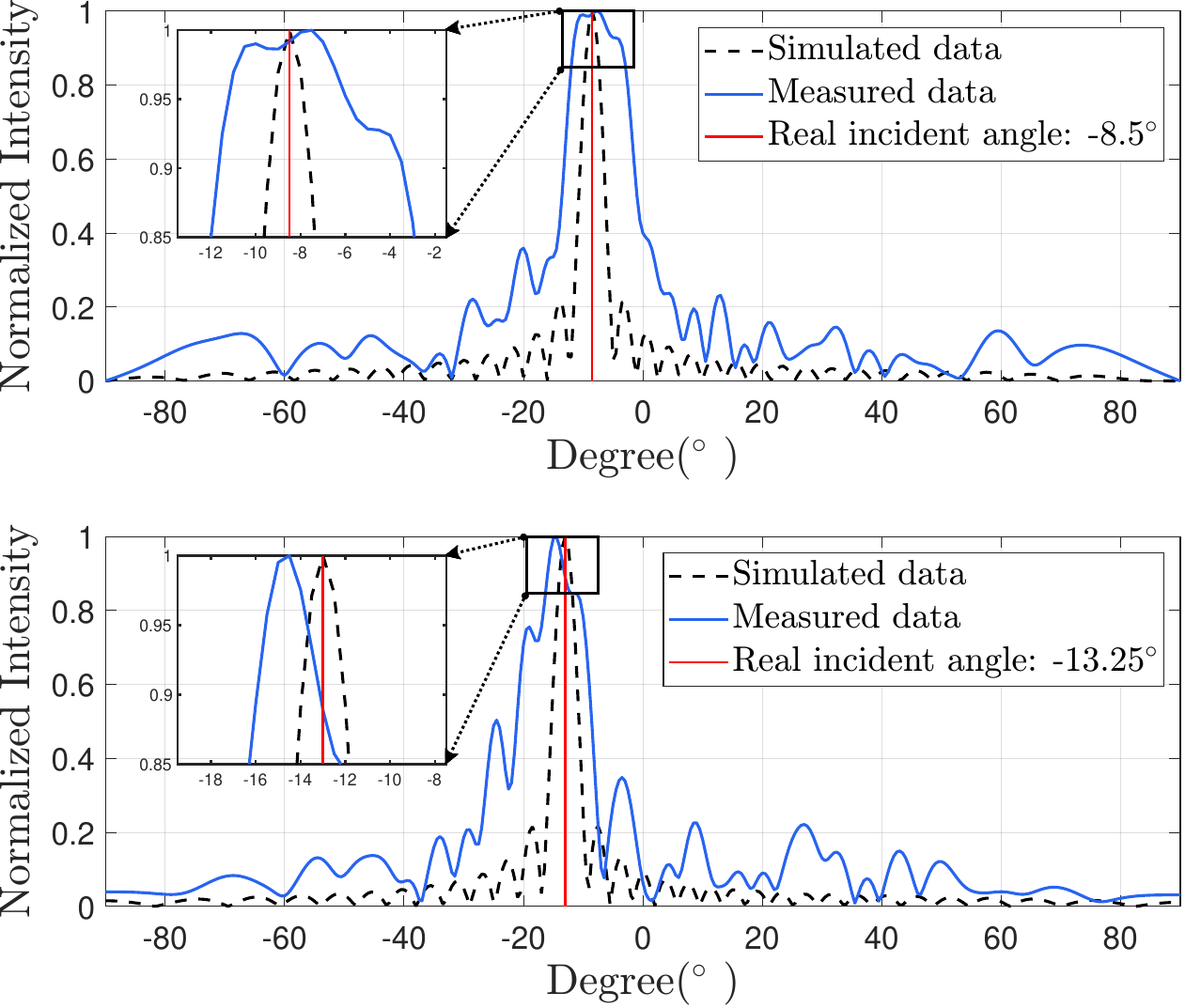}}
  \caption{Results of DoA derived from both simulated and measured data by the magnitude-only reconstruction algorithm. The first row displays the results of DoA from the left RIS. The second row displays the results of DoA from the right RIS. Results demonstrate high accuracy with errors within $2^\circ$.}
  \label{1simulation} 
\end{figure}

\begin{figure}[htbp] 
  \centerline{\includegraphics[width=1\columnwidth]{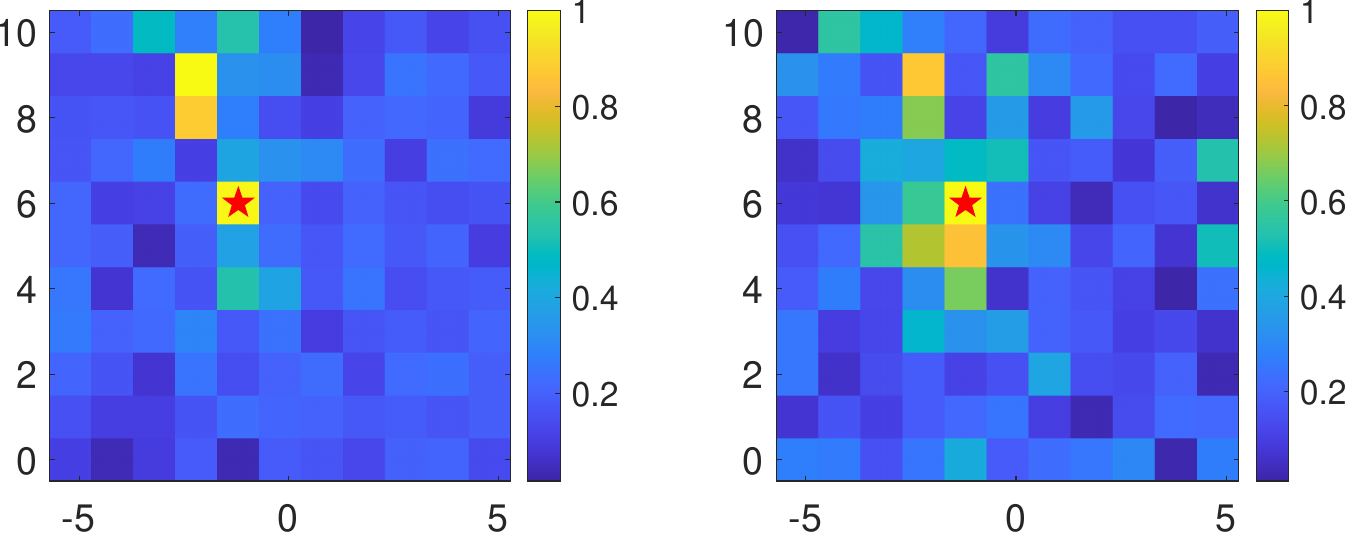}}
  \caption{Recovered results using the magnitude-only reconstruction algorithm in the chamber. The left figure displays the reconstructed scene from simulated data. The right figure shows the reconstructed scene from measured data.} 
  \label{experiment_Reconstructed_real1}
\end{figure}

It's noted that the validation using USRP has a limitation: while the strength of the carrier wave is easy to measure, accurately capturing the phase is more challenging. Therefore, the backward sensing process must rely on magnitude-only reconstruction algorithms, which are also applicable to other wireless devices such as Wi-Fi or spectrum analyzers. Mathematically, the magnitude-only backward sensing problem is formulated as follows: given $\mathbf{H}$ and $\lvert \mathbf{S} \rvert$, find $\mathbf{E}$ such that $\lvert \mathbf{S} \rvert \approx \lvert \mathbf{H} \mathbf{E} \rvert $. 
In contrast to phased measurements described by \eqref{E:LinearMeasurement}, the missing phase makes the problem more complex. The phaseless reconstruction problem is highly nonlinear and significantly more ill-posed than the phased reconstruction problem~\cite{ammari2016phased}. Various approaches have been developed to address the phaseless reconstruction problem, including gradient descent methods, iterative projection methods, and convex optimization-based methods~\cite{fannjiang2020numerics}. We focus on validating our backward sensing approach rather than delving into the specifics of these algorithms. For our purposes, we employ the reweighted Wirtinger flow (RWF) algorithm as the solver~\cite{yuan2017phase}.

We initiate the validation of the RWF algorithm by estimating the DoA using data independently obtained from each of the two RISs. Fig.~\ref{1simulation} shows the DoA estimations from both RISs. Specifically, a single source is localized at $-8.5^\circ$ for the left RIS and $-13.25^\circ$ for the right RIS. In numerical experiments, a distinct peak is observed with DoA estimations of $-8.5^\circ$ and $-13^\circ$ for the left RIS and the right RIS, respectively. In the results estimated using measured data, the DoA estimations are $-7.5^\circ$ for the left RIS and $-14.5^\circ$ for the right RIS, with errors within two degrees. This level of accuracy is comparable to prior reports on RIS-aided or low-complexity DoA estimation systems~\cite{lin2021single,11072249}. These figures illustrate the remarkably accurate estimation of incident angles, thereby demonstrating the accuracy of our modeling and the effectiveness of the RWF algorithm within this experimental setup.

To further validate the collaborative sensing capabilities of the two RISs, we define the RoI as a square with a size of 12 m $\times$ 12 m, discretized into a grid of 12 $\times$ 12, resulting in each pixel being 1 m $\times$ 1 m. The final results, illustrated in Fig.~\ref{experiment_Reconstructed_real1}, provide a comprehensive overview. The true location of the source is indicated by red pentagrams. The plots clearly demonstrate that the power source can be precisely located using the two RISs. These findings validate the feasibility of the proposed backward sensing approach with multiple RISs.

\begin{figure*}
  \begin{equation} \label{covariance} 
    \begin{aligned}
      {\mathbf{V}}^*(\Theta) \mathbf{V}(\Theta) 
      &= \begin{bmatrix}
        1 & e^{ -j 2 \pi d \sin \theta^{\text{i}} / \lambda }  & \cdots & e^{ -j 2 \pi (N-1) d \sin \theta^{\text{i}} / \lambda }\\
        1 & e^{ -j 2 \pi d \sin ( \theta^{\text{i}}+\Delta ) / \lambda }  & \cdots & e^{ -j 2 \pi (N-1) d \sin ( \theta^{\text{i}}+\Delta ) / \lambda } \\
      \end{bmatrix} 
      \begin{bmatrix}
          1      & 1 \\
          e^{ j 2 \pi d \sin \theta^{\text{i}} / \lambda }    & e^{ j 2 \pi d \sin ( \theta^{\text{i}}+\Delta ) / \lambda }\\
          \vdots & \vdots\\
          e^{ j 2 \pi (N-1) d \sin \theta^{\text{i}}/ \lambda }  & e^{ j 2 \pi (N-1) d \sin ( \theta^{\text{i}}+\Delta ) / \lambda }
        \end{bmatrix}\\
    &= \begin{bmatrix}
      N       & 1  + \cdots + e^{ j 2 \pi (N-1) d  (\sin ( \theta^{\text{i}}+\Delta ) -\sin \theta^{\text{i}})/ \lambda }\\
      1  + \cdots + e^{ j 2 \pi (N-1) d (\sin \theta^{\text{i}}- \sin ( \theta^{\text{i}}+\Delta )) / \lambda } & N
    \end{bmatrix} .
    \end{aligned}
  \end{equation}
  
  \begin{equation} \label{lambda1}
    \begin{aligned}
      \lambda_1,\lambda_2  
      & = N \pm \sqrt{(1 + \dots + e^{ j 2 \pi (N-1) d (\sin \theta^{\text{i}}- \sin ( \theta^{\text{i}}+\Delta )) / \lambda } ) (1+\dots+e^{ j 2 \pi (N-1) d  (\sin ( \theta^{\text{i}}+\Delta ) -\sin \theta^{\text{i}})/ \lambda })}\\
      & = N \pm \left|\frac{1-e^{j 2 \pi N d (\sin \theta^{\text{i}}- \sin ( \theta^{\text{i}}+\Delta ))\lambda}}{1-e^{j 2 \pi d (\sin \theta^{\text{i}}- \sin ( \theta^{\text{i}}+\Delta )) / \lambda}} \right| 
      = N \pm \left| \frac{ \sin ( \pi N d (\sin \theta^{\text{i}} - \sin ( \theta^{\text{i}}+\Delta ) ) / \lambda ) }{ \sin ( \pi d ( \sin \theta^{\text{i}} - \sin ( \theta^{\text{i}}+\Delta ) ) / \lambda ) } \right| .\\
    \end{aligned}
  \end{equation}
  \medskip
  \hrule
\end{figure*}

\section{Conclusion} \label{S:Conclusion}

This paper investigates the feasibility of RIS-centric backward sensing. We demonstrate that a single RIS enables DoA estimation, while the spatial diversity provided by multiple RISs facilitates multi-source localization. A theoretical framework based on singularity analysis is developed to quantify key performance indicators, and numerical simulations validate the analysis across diverse configurations. Finally, a USRP-based proof-of-concept prototype provides experimental validation of RIS-centric sensing. The current prototype explores single-target localization in a controlled line-of-sight environment. Future work will extend the system to multi-target localization under non-line-of-sight and dynamic conditions~\cite{li2023riscan}, and investigate optimization-based or AI-driven strategies for adaptive sensing and robust parameter estimation in complex indoor environments.

\section{Acknowledgment} \label{S:Acknowledgment}
The authors gratefully acknowledge the editor and anonymous reviewers for their constructive suggestions, which significantly enhanced the rigor and clarity of this work.


\appendices

\section{Proof of Lemma~\ref{L:4}}\label{S:Appendix_A3}
\begin{proof} 
For the matrix $\mathbf{V} ( \mathbf{\Theta}^\text{i})$ defined by \eqref{E:VandermodeMatrix}, we first calculate $\mathbf{V}^* ( \mathbf{\Theta}^\text{i}) \mathbf{V} ( \mathbf{\Theta}^\text{i})$, as shown in 
\eqref{covariance}. It is straightforward to calculate the eigenvalues of $\mathbf{V}^* ( \mathbf{\Theta}^\text{i}) \mathbf{V} ( \mathbf{\Theta}^\text{i})$, which are given in \eqref{lambda1}. Finally, we establish the singular values of $\mathbf{V} ( \mathbf{\Theta}^\text{i})$. 
\end{proof}

\section{Proof of Lemma~\ref{L:5}}\label{S:Appendix_A4}
\begin{proof} 
The main idea is to find the series expansion
\begin{equation} \label{PL5:1}
  \frac{\sin N x}{\sin x} = \sum_{k=0}^{\infty} a_{k} x^{2 k} .
\end{equation}
It is easy to verify that
\begin{equation} \label{D:4}
  \begin{aligned}
    \frac{\sin N x}{\sin x}
    &= \frac{e^{i N x}-e^{-i N x}}{e^{i x}-e^{-i x}}
    =\sum_{l=0}^{N-1} e^{i(N-1-2 l) x}\\
    &=\left\{\begin{array}{ll}
      1+2 \sum_{l=1}^{\frac{N - 1}{2}} \cos 2 l x  & N \text { odd}, \\
      2 \sum_{l=1}^{\frac{N}{2}} \cos (2 l -1) x  & N \text { even}.
    \end{array} \right.
  \end{aligned}
\end{equation}
Using the Taylor expansion $\cos x = \sum_{k=0}^{\infty} \frac{(-1)^k}{(2k)!} x^{2 k}$, we find 
\begin{equation}\label{E:TaylorCos1}
  \cos 2 l x = \sum_{k=0}^{\infty} \frac{(-1)^k}{(2k)!} (2 l)^k x^{2 k} ,
\end{equation}
and
\begin{equation}\label{E:TaylorCos2}
  \cos (2 l -1) x = \sum_{k=0}^{\infty} \frac{(-1)^k}{(2k)!} (2 l - 1)^k x^{2 k} .
\end{equation}
Substituting \eqref{E:TaylorCos1} and \eqref{E:TaylorCos2} into \eqref{D:4} yields
\begin{equation}
    a_{k} = \frac{(-1)^{k}}{(2 k) !} \sum_{l=0}^{N-1}(N-1-2 l)^{2 k} .
\end{equation}
It follows that $a_0 = N$ and $a_1 = -\frac{N (N^2-1)}{6}$. Finally, for sufficiently small $x$, we have 
\begin{equation}
  \frac{ \sin (N x) }{\sin x} \approx N - \frac{ N ( N^2-1 ) x^2 }{ 6 }. 
\end{equation}
\end{proof}

\bibliographystyle{IEEEtran}
\bibliography{Reference}

\begin{IEEEbiography}[{\includegraphics[width=1in,height=1.25in,clip,keepaspectratio]{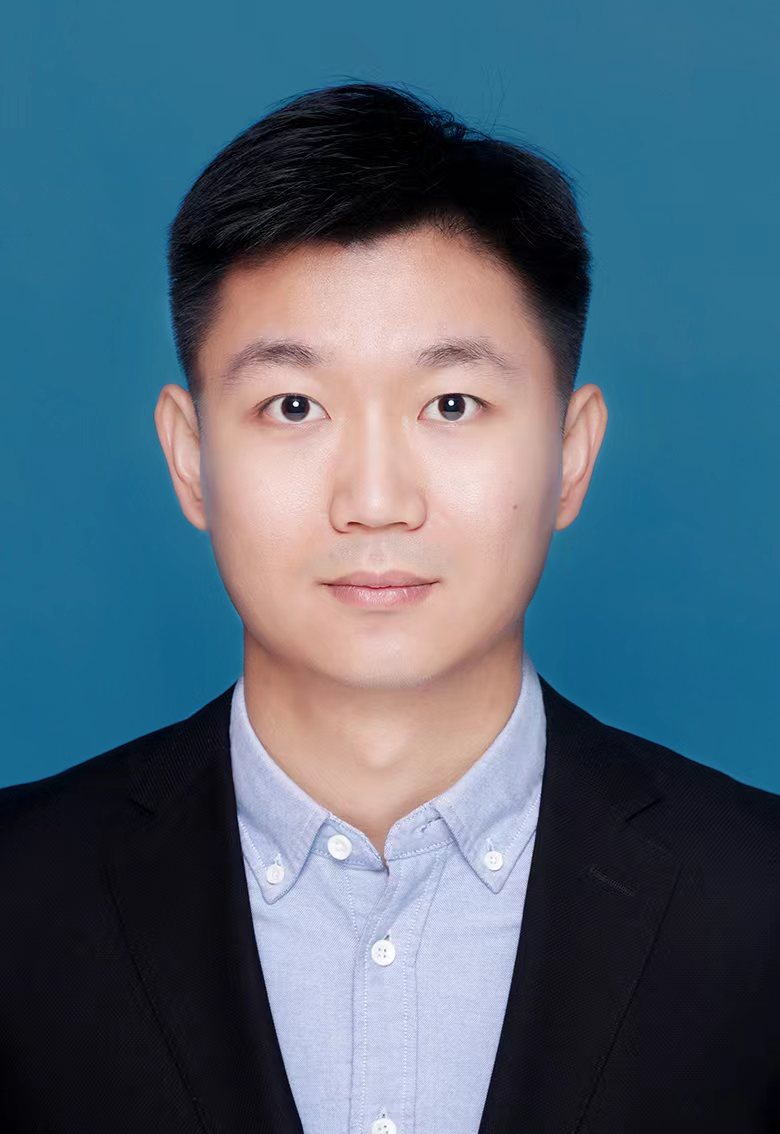}}]{Fuhai Wang}
(Graduate Student Member, IEEE) received the B.S. degree in communication engineering from Chongqing University of Posts and Telecommunications, China, in 2018, and the M.S. degree in materials science from University of Chinese Academy of Sciences, China, in 2021. He is currently pursuing the Ph.D. degree with the Institute of Artificial Intelligence, Huazhong University of Science and Technology, China. His current research interests include wireless localization and sensing, metamaterials imaging, and RF-based scene reconstruction. \end{IEEEbiography}

\begin{IEEEbiography}[{\includegraphics[width=1in,height=1.25in,clip,keepaspectratio]{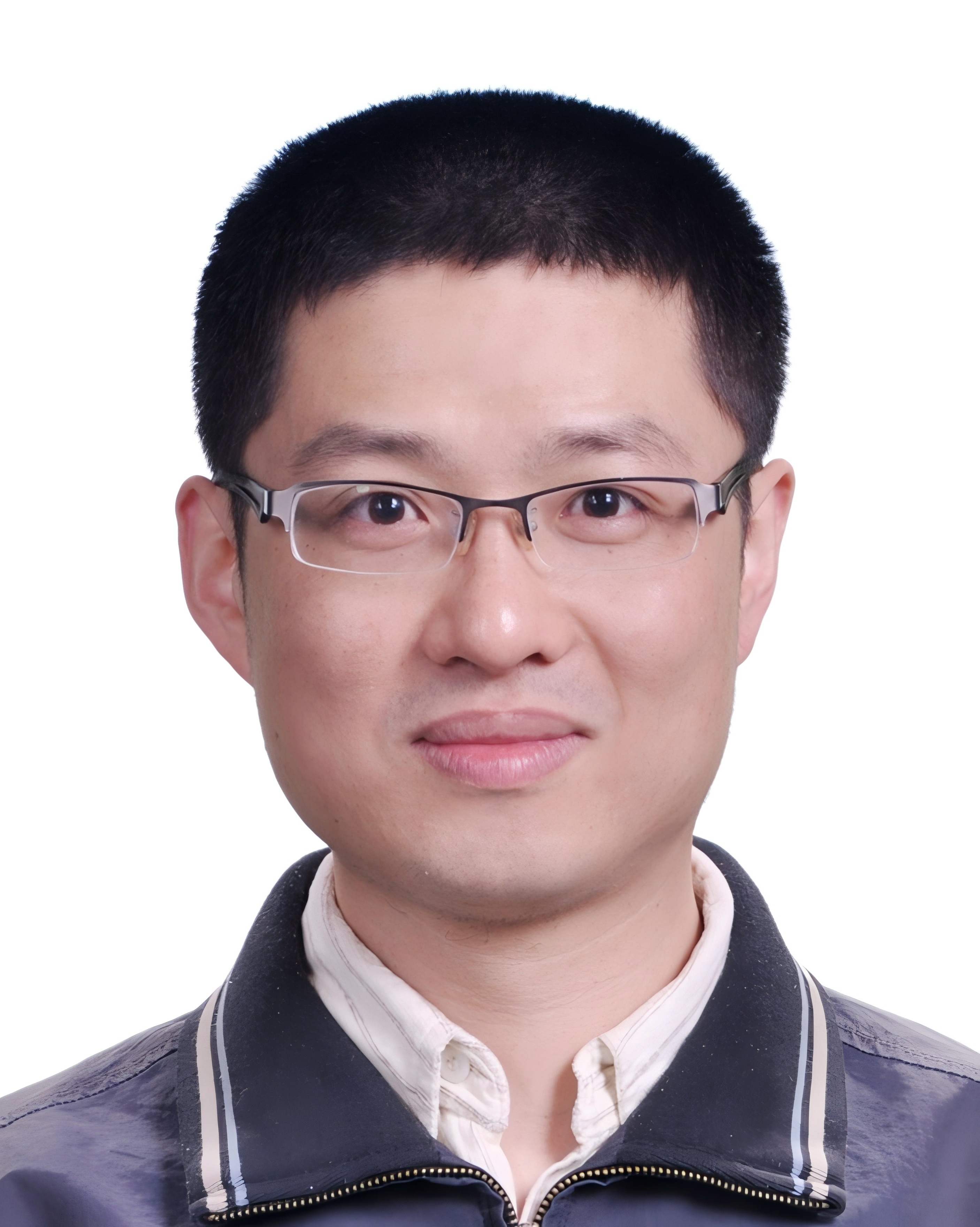}}]{Tiebin Mi}
(Member, IEEE) received the B.E. degree in computer science from Xidian University, China, in 2002, and the Ph.D. degree in electrical engineering from the Institute of Acoustics, Chinese Academy of Sciences, China, in 2010. Currently, he is a Lecturer (an Assistant Professor) with the School of Electronic Information and Communications, Huazhong University of Science and Technology, China. His current research interests include wireless communications, high-dimensional signal processing, random matrix theory, and reconfigurable intelligent surfaces.\end{IEEEbiography}

\begin{IEEEbiography}[{\includegraphics[width=1in,height=1.25in,clip,keepaspectratio]{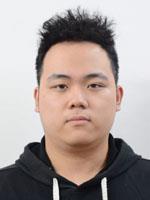}}]{Chun Wang}
received the B.S. degree in electromagnetic fields and wireless technology from Huazhong University of Science and Technology. He is currently pursuing the M.S. degree with the Huazhong University of Science and Technology. His current research interests include wireless localization and sensing, dynamic metamaterial antenna.
\end{IEEEbiography} 

\begin{IEEEbiography}[{\includegraphics[width=1in,height=1.25in,clip,keepaspectratio]{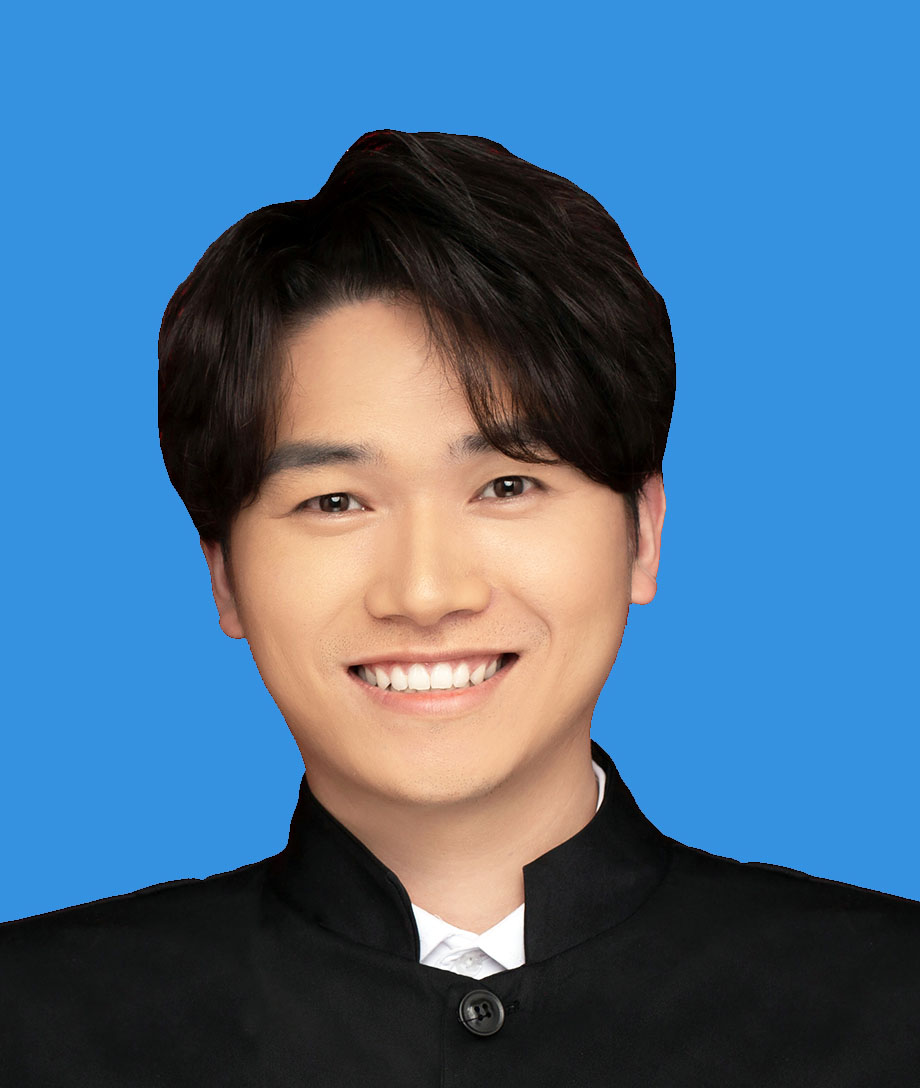}}]{Rujing Xiong}
(Member, IEEE) received the B.S. degree in Bioinformatics from Zhengzhou University, China, in 2017, the M.Eng. degree in Electronics and Communication Engineering from Central South University, China, in 2020, and the Ph.D. degree in Information and Communication Engineering from Huazhong University of Science and Technology, China. He is currently a Research Fellow with the School of Science and Engineering, The Chinese University of Hong Kong, Shenzhen, China. His research interests include wireless communications, signal processing, non-convex optimization, reconfigurable intelligent surfaces, and analog computation. He received the Best Paper Award for IEEE/CIC ICCC 2025.
\end{IEEEbiography}

\begin{IEEEbiography}[{\includegraphics[width=1in,height=1.25in,clip,keepaspectratio]{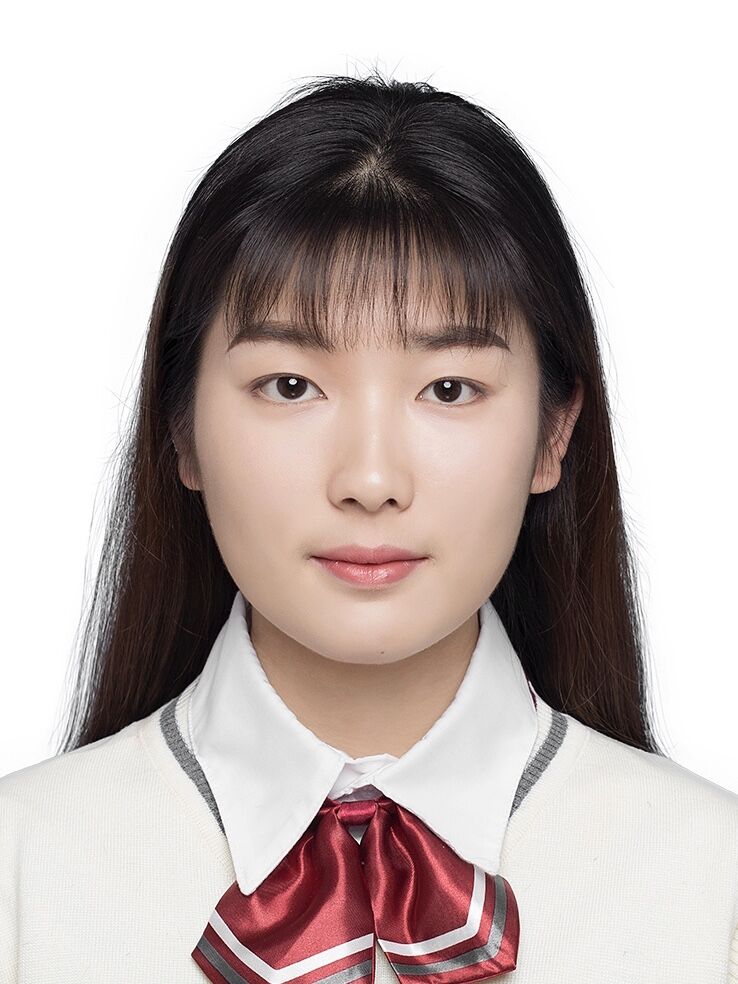}}]{Zhengyu Wang}
(Student Member, IEEE) received the B.S. degree in communication engineering from Central South University, Changsha, China, in 2020.  She is currently pursuing the Ph.D. degree in Huazhong University of Science and Technology, Wuhan, China. Her current research interests include Reconfigurable Intelligent Surfaces (RISs), Random Matrix Theory.
\end{IEEEbiography}

\begin{IEEEbiography}[{\includegraphics[width=1in,height=1.25in,clip,keepaspectratio]{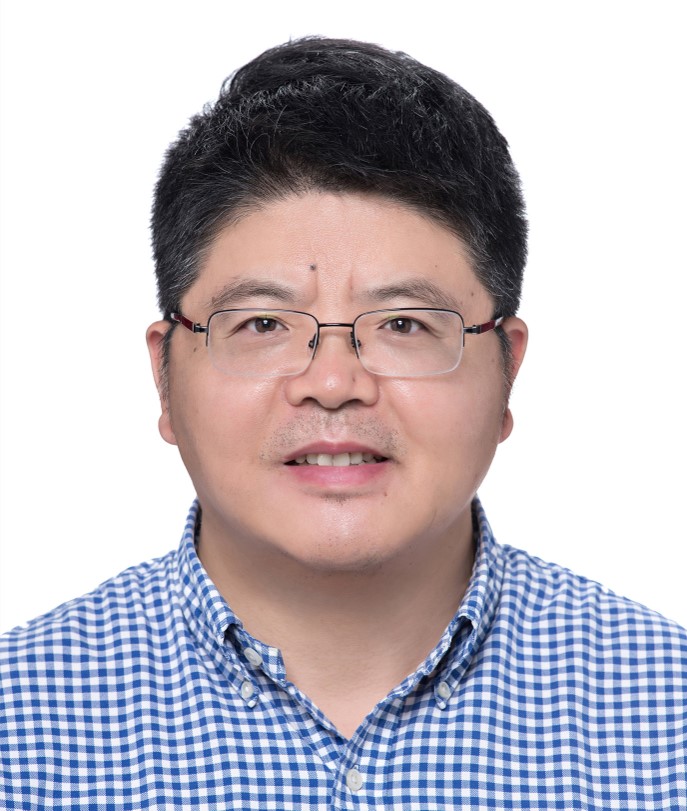}}]{Robert Caiming Qiu}
(Fellow, IEEE) received the Ph.D. degree in electrical engineering from New York University (former Polytechnic University) Brooklyn, NY, USA. He joined the School of Electronic Information and Communications, Huazhong University of Science and Technology, China, as a Full Professor, in 2020. Before joining HUST, he was an Associate Professor with the Department of Electrical and Computer Engineering, Center for Manufacturing Research, Tennessee Tech University, Cookeville, TN, USA, in 2003, where he became a Professor in 2008. He founded and served as the CEO and President of Wiscom Technologies Inc., a company specializing in manufacturing and marketing WCDMA chipsets. In 2003, he was acquired by Intel. He was with GTE Laboratories Inc. (now Verizon), Waltham, MA, USA, and Bell Laboratories, Lucent, Whippany, NJ, USA. He has coauthored Cognitive Radio Communication and Networking: Principles and Practice (John Wiley, 2012) and Cognitive Networked Sensing: A Big Data Way (Springer, 2013) and authored Big Data and Smart Grid (John Wiley, 2015). He has authored over 100 journal articles/book chapters and 120 conference papers. He was a Guest Book Editor of Ultra-Wideband (UWB) Wireless Communications (New York: Wiley, 2005) and three Special Issues on UWB, including the IEEE Journal on Selected Areas in Communications, IEEE Transactions on Vehicular Technology, and IEEE Transactions on Smart Grid. Furthermore, he has made 15 contributions to 3GPP and IEEE standards bodies. He has served as a TPC Member for GLOBECOM, ICC, WCNC, and MILCOM. He has also served as an Associate Editor for IEEE Transactions on Vehicular Technology and other international journals. \end{IEEEbiography}

\end{document}